\documentclass{llncs}
\usepackage{multirow}
\usepackage{amsmath}
\usepackage{xcolor}
\usepackage{graphicx}
\usepackage{gastex}
\usepackage{amssymb}

\def\A{\mathcal{A}}
\def\phi{\varphi}
\def\theta{\vartheta}
\def\true{\mathsf{True}}
\newcommand{\tran}[1]{\xrightarrow[]{#1}}
\def\lz#1{\color{red}LZ: \texttt{#1} :ZL\color{black}}

\title{On the Relationship between  LTL Normal Forms and B\"uchi Automata}
\titlerunning{On the relationship between LTL Normal Forms and B\"uchi Automata} 
\author{Jianwen Li\inst{1} \and
Geguang Pu\inst{1} \and
Lijun Zhang\inst{2} \and \\
Zheng Wang\inst{1} \and
Jifeng He\inst{1}\and
Kim G. Larsen\inst{3}
}

\institute
{
 \inst{}%
 Programming, Logic, and Semantics Group\\
Software Engineering Institute East China Normal University, 
P. R.China
 \and
 \inst{}%
DTU Informatics, Technical University of Denmark
\and
\inst{}%
Computer Science, Aalborg University, Denmark
}

\begin{document}

\maketitle

\begin{abstract}

  In this paper, we consider the problem of translating LTL formulas
  to B\"{u}chi automata. We first translate the given LTL formula into
  a special \emph{disjuctive-normal form} (DNF). The formula will be
  part of the state, and its DNF normal form specifies the atomic
  properties that should hold immediately (labels of the transitions)
  and the \emph{formula} that should hold afterwards (the
  corresponding successor state). Surprisingly, if the given formula is
  Until-free or Release-free, the B\"uchi automaton can be obtained
  directly in this manner.  For a general formula, the construction is
  slightly involved: an additional component will be needed for each
  formula that helps us to identify the set of accepting
  states. Notably, our construction is an on-the-fly construction, and
  the resulting B\"uchi automaton has in worst case $2^{2n+1}$ states
  where $n$ denotes the number of subformulas. Moreover, it has a
  better bound $2^{n+1}$ when the formula is Until- (or Release-) free.

\end{abstract}

\section{Introduction}\label{introduction}
Translating Linear Temporal Logic (LTL) formulas to their equivalent
automata (usually B\"{u}chi automata) has been studied for nearly
thirty years.
This translation plays a key role in the automata-based model
checking~\cite{Vardi86}: here the automaton of the negation of the LTL
property is first constructed, then the verification process is
reduced to the emptiness problem of the product.  Gerth et
al.~\cite{Gerth95} proposed an on-the-fly construction approach to
generating B\"{u}chi automata from LTL formulas, which means that the
counterexample can be detected even only a part of the property
automaton is generated. They called it a tableau construction
approach, which became widely used and many subsequent
works~\cite{Somenzi00,Giannakopoulou02,Daniele99,Etessami00,TACAS12}
for optimizing the automata under construction are based on
it.

In this paper, we propose a novel construction by making use of the
notion of \emph{disjuctive-normal forms} (DNF). For an LTL formula
$\phi$, its DNF normal form is an equivalent formula of the form
$\bigvee _i(\alpha_i\wedge X \phi_i)$ where $\alpha_i$ is a finite
conjunction of literals (atomic propositions or their negations), and
$\phi_i$ is a conjunctive LTL formula such that the root operator of it is not a
disjunction. We show that any LTL formula can be transformed into an
equivalent DNF normal form, and refer to $\alpha_i\wedge X\phi_i$ as a
clause of $\phi$. It is easy to see that any given LTL formula induces
a labelled transition system (LTS): states correspond to formulas, and we
assign a transition from $\phi$ to $\phi_i$ labelled with $\alpha_i$,
if $\alpha_i\wedge X\phi_i$ appears as a part of the DNF form of
$\phi$. Figure~\ref{ltl_ts_and_ba} demonstrates our idea in which the
transition labels are omitted.

\begin{figure}
\begin{minipage}[t]{0.5\linewidth}
\centering
\includegraphics[scale = 0.5]{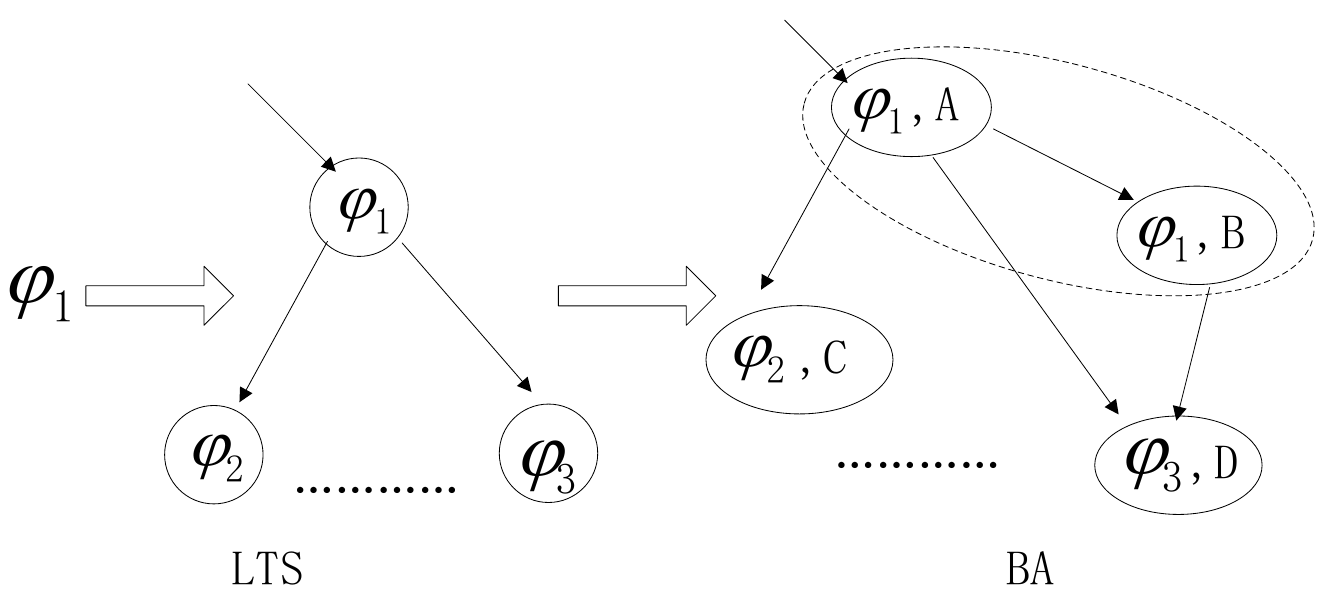}\\
\caption{A demonstration of our idea}\label{ltl_ts_and_ba}
\end{minipage}
\begin{minipage}[t]{0.5\linewidth}
\centering
\includegraphics[scale=0.9]{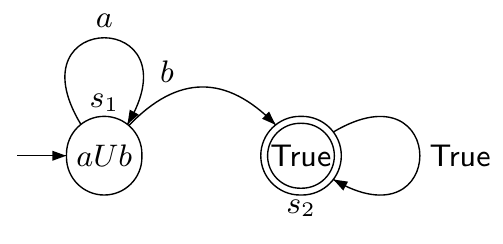}
\caption{The B\"uchi automaton for $aUb$.}\label{fig:buchiforaUb}
\end{minipage}
\end{figure}


The LTS is the starting point of our construction. Surprisingly, for
Until-free (or Release-free) formulas, the B\"uchi automaton can be
obtained directly by equipping the above LTS with the set of accepting
states, which is illustrated as follows. Consider the formula $aUb$,
whose DNF form is $(b\wedge X(\mathsf{True}))\vee (a\wedge
X(aUb))$. The corresponding B\"{u}chi automaton for $aUb$ is shown in
Figure~\ref{fig:buchiforaUb} where nodes $aUb$ and $\true$ represent
formulas $aUb$ and $\mathsf{True}$ respectively. The transitions are
self-explained. By semantics, we know that if the run $\xi$ satisfies
a Release-free formula $\phi$, then there must be a finite satisfying
prefix $\eta$ of $\xi$ such that any paths starting with $\eta$
satisfy $\phi$ as well. Thus, for this class of formulas, the state
corresponding to the formula $\true$ is considered as the single
accepting state. The Until-free formulas can be treated in a similar
way by taking the set of all states as accepting.


The main contribution of the paper is to extend the above construction
to general formulas. As an example we consider the formula $\psi=G(a U
b)$, which has the normal form $ (b\wedge X\psi) \vee ( a\wedge
X(aUb\wedge\psi))$. Note here the formula $\true$ will be even not
reachable. The most challenging part of the construction will then be
identification of the set of accepting states. For this purpose, we
identify subformulas that will be reached infinitely often, which we
call looping formulas. Only some of the looping formulas contribute to
the set of accepting states. These formulas will be the key to our
construction: we characterize a set of atomic propositions for each
formula, referred to as the \emph{obligation set}. The set contains
properties that must occur infinitely often to make the given formula
satisfiable. In our construction, we add an additional component to
the states to keep track of the obligations, and then define accepting states
based on it -- an illustrating example can be found in Section
\ref{example}.

Our construction for general formula has at most $2^{2n+1}$ states
with $n$ denoting the number of subformulas. The number of states for
the Release/Until cases is bounded $2^{n+1}$. Recall the complexity of
$2^{O(n)}$ \cite{Gerth95} of the classical tableau construction. To
the best of our knowledge, this is the first time that one can give a
precise bound on the exponent for such construction.

\subsubsection*{Related Work}
As we know, there are two main approaches to B\"{u}chi automata
construction from LTL formulas. The first approach generates the
alternating automaton from the LTL formula and then translates it to
the equivalent B\"{u}chi automaton~\cite{Vardi96}. Gastin et
al.~\cite{Gastin01} proposed a variant of this construction in 2001,
which first translates the very weak alternating co-B\"uchi automaton
to generalised automaton with accepting transitions which is then
translated into B\"uchi automaton.  In particular, the experiments
show that their algorithm outperforms the others if the formulas under
construction are restricted on fairness conditions. Recently
Babiak et al.~\cite{TACAS12} proposed some optimization strategies
based on the work~\cite{Gastin01}.  

The second approach was proposed in 1995 by Gerth et
al.~\cite{Gerth95}, which is called the \textit{tableau}
construction. This approach can generate the automata from LTL
on-the-fly, which is widely used in the verification tools for
acceleration of the automata-based verification process.
Introducing the (state-based) \textit{Generalized B\"{u}chi
Automata} (GBA) is the important feature for the tableau
construction. Daniele et al.~\cite{Daniele99} improved the tableau
construction by some simple syntactic techniques. Giannakopoulou and
Lerda~\cite{Giannakopoulou02} proposed another construction approach
that uses the transition-based Generalized B\"{u}chi automaton
(TGBA).
And some optimization techniques~\cite{Etessami00,Somenzi00}  have been
proposed to reduce the size of the generated automata. For instance,
Etessami and Holzmann~\cite{Etessami00} described the
optimization
techniques including proof theoretic reductions (formulas
rewritten), core algorithm tightening and the automata theoretic
reductions (simulation based).

\subsubsection*{Organization of the paper.}
Section \ref{example}
illustrates our approach by a running example. Section \ref{sec:pre}
introduces preliminaries of B\"uchi automata and LTL formulas and then introduces the \emph{disjunctive-normal form} for LTL formulas; Section~\ref{construction} specifies the
proposed \emph{DNF-based } construction; Section~\ref{discussion}
discusses how our approach is related to the tableau
construction in~\cite{Gerth95}. 
Section~\ref{conclusion} concludes the paper.

\section{A Running Example}~\label{example} We consider the formula
$\phi_1=\textrm{G}(bUc\wedge dUe)$ as our running example. The DNF
form of $\phi_1$ is given by:
\begin{align*}
\phi_1=
(c\wedge e\wedge X(\phi_1) )
\vee (b\wedge e\wedge X(\phi_2))
\vee (c\wedge d\wedge X(\phi_3))
\vee (b\wedge d\wedge X(\phi_4)
\end{align*}
where $\phi_2=bUc\wedge G(bUc\wedge dUe)$, $\phi_3=dUe\wedge
G(bUc\wedge dUe)$, $\phi_4=bUc\wedge dUe \wedge G(bUc\wedge dUe)$.  It
is easy to check that the above DNF form is indeed equivalent to
formula $\phi_1$.  Interestingly, we note that
$\phi_1,\phi_2,\phi_3,\phi_4$ all have the same DNF form above.

\begin{figure}
\centering
  \includegraphics[scale=1.0]{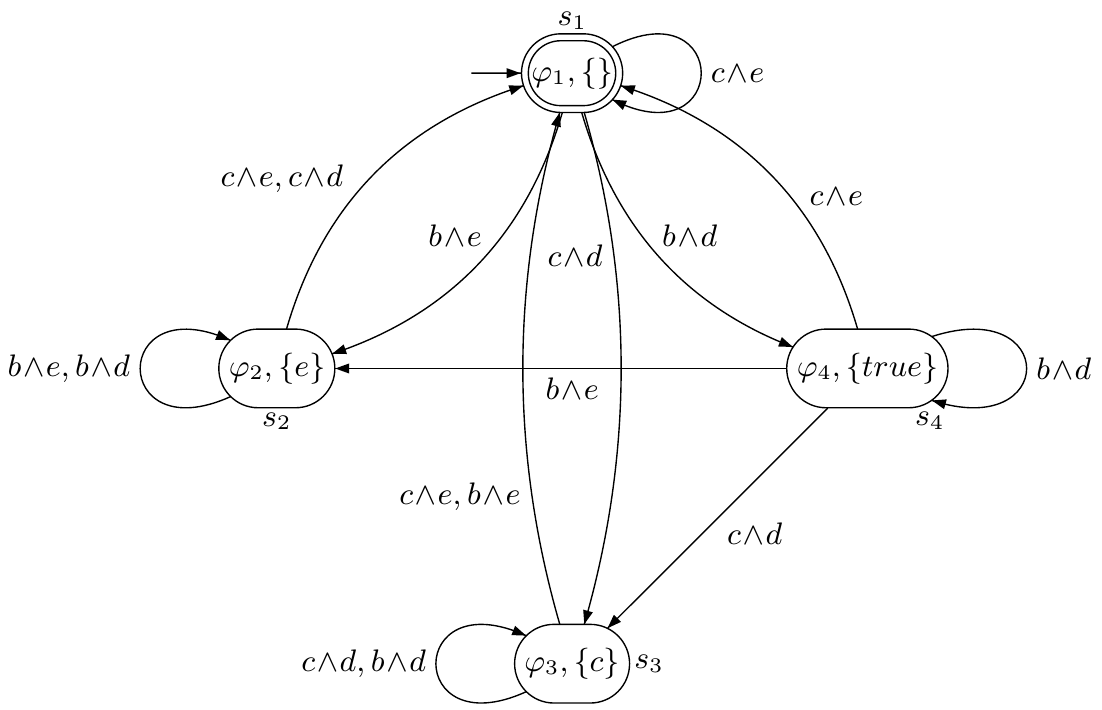}\\
  \caption{\label{fig:buechi}
The B\"uchi automaton for the formula $\phi_1$.}
\end{figure}

The corresponding B\"uchi automaton for $\phi_1$ is
depicted in Fig.~\ref{fig:buechi}.
We can see that there are four states in the generated automata,
corresponding to the four formulas $\phi_i (\mathit{i = 1, 2, 3,
  4})$. The state corresponding to the formula $\phi_1$ is also the
initial state. The transition relation is obtained by observing the
DNF forms: for instance we have a self-loop for state $s_1$ with label
$c\wedge e$. If we observe the normal form of $\phi_1$, we can see
that there is a term ($c\wedge e\wedge X(\phi_1)$), where there is a
conjunction of two terms $c\wedge e$ and $X(\phi_1)$, and $\phi_1$ in
$X$ operator corresponds to the node $s_1$ and $c\wedge e$ corresponds
the loop edge for $s_1$. 

Thus, the \emph{disjunctive-normal form} of the formula has a very
close relation with the generated automaton. The most difficult part
is to determine the set of accepting states of the automaton. We give
thus here a brief description of several notions introduced for this
purpose in our running example.  The four of all the formulas
$\mathit{\phi_i(i=1,2,3,4)}$ have the same \textit{obligation set},
i.e. $OS_{\phi_i}=\{\{c,e\}\}$, which may vary for different formulas.
In our construction, every \textit{obligation} in the
\textit{obligation set} of each formula identities the properties
needed to be satisfied infinitely if the formula is satisfiable. For
example, the formulas $\phi_i(\mathit{i=1,2,3,4})$ are satisfied if
and only if all properties in the obligation $\{c,e\}$ are met
infinitely according to our framework. Then, a state consists of a
formula and the \textit{process set}, which records all the properties
that have been met so far.  For simplicity, we initialize the
\textit{process set} $P_1$ of the initial state $s_1$ with the empty
set. For the state $s_2$, the corresponding process set $P_2 = \{e\}$
is obtained by taking the union of $P_1$ and the label $\{b,e\}$ from
$s_1$. The label $b$ will be omitted as it is not contained in the
obligation.  Similarly one can conclude $P_3 = \{c\}$ and $P_4 =
\{true\}$: here the property $true$ implies no property has been met
so far. When there is more than one property in the \textit{process
  set}, the $\{true\}$ can be erased, such as that in state
$s_3$. Moreover, the \textit{process set} in a state will be reset to
empty if it includes one \textit{obligation} in the formula's
\textit{obligation set}. For instance, the transition in the figure
$s_2\tran{c\wedge d}s_1$ is due to that $P_1' = P_2\cup
\{c\}=\{c,e\}$, which is actually in $OS_{\phi_1}$. So $P_1'$ is reset
to the empty set. One can also see the same rule when the transitions
$s_2\tran{c\wedge e}s_1$, $s_4\tran{c\wedge e}s_1$, $s_3\tran{b\wedge
  e}s_1$ occur.



Through the paper, we will go back to this example again when we
explain our construction approach.

\section{B\"uchi Automaton, LTL and Disjunctive Normal Form}\label{sec:pre}

\subsection{B\"uchi Automaton}\label{sec:automata}

A B\"{u}chi automaton is a tuple $\mathcal{A}=(S, \Sigma, \delta ,
S_0, F)$, where $S$ is a finite set of states, $\Sigma$ is a finite set of
alphabet symbols , $\delta: S\times \Sigma \to 2^S$ is the transition
relation, $S_0$ is a set of initial states, and $F\subseteq S$ is a
set of accepting states of $\mathcal{A}$.

  We use $w, w_0\in\Sigma$ to denote alphabets in $\Sigma$, and $\eta, \eta_0\in
  \Sigma^*$ to denote finite sequences.  A \emph{run}
  $\xi=w_0w_1w_2\ldots$ is an infinite sequence over
  $\Sigma^\omega$. For $\xi$ and $k\geq 1$ we use $\xi ^k=w_0w_1\ldots
  w_{k-1}$ to denote the prefix of $\xi$ up to its $k$th element (the
  $k+1$th element is not included) as well as $\xi _k$ to denote the
  suffix of $w_kw_{k+1}\ldots$ from its $(k+1)$th element (the $k+1$th
  element is included).  Thus, $\xi=\xi^k \xi_k$. For notational
  convenience we write $\xi_0=\xi$ and $\xi^0=\varepsilon$
  ($\varepsilon$ is the empty string). The run $\xi$ is accepting if
  it runs across one of the states in $F$ infinitely often.  

 \subsection{Linear Temporal Logic}\label{sec:ltl}
 We recall the linear temporal logic (LTL) which is widely used as a
 specification language to describe the properties of reactive
 systems. Assume $AP$ is a set of atomic properties, then the syntax of
 LTL formulas is defined by:
\begin{align*}
\phi\ ::=\ a \mid \neg a \ |\ \phi\wedge \phi\ |\ \phi\vee \phi\ |\ \phi\ U \phi\ |\ \phi\ R\ \phi\ |\ X\ \phi
\end{align*}
where $a\in AP$, $\phi$ is an LTL formula.  We say $\phi$ is a \emph{literal} if it is a
proposition or its negation. In this paper we use lower case letters to
denote atomic properties and $\alpha$, $\beta$, $\gamma$ to denote
propositional formulas (without temporal operators), and use $\phi$, $\psi$, $\theta$, $\mu$, $\nu$
and $\lambda$ to denote LTL formulas.

Note that w.l.o.g.  we are considering LTL formulas in negative normal
form (NNF) -- all negations are pushed down to literal level.  LTL
formulas are interpreted on infinite sequences (correspond to runs of the automata) $\xi\in \Sigma
^\omega$ with $\Sigma =2^{AP}$. The Boolean connective case is trivial,
and the semantics of temporal operators is given by:
\begin{itemize}
\item $\xi \models \phi_1\ U\ \phi_2$ iff there exists $i\geqslant 0$
  such that $\xi_i\models \phi_2$ and for all $0 \leqslant j < i,
  \xi_j\models \phi_1$;
\item $\xi \models \phi_1\ R\ \phi_2$ iff either
  $\xi_i\vDash\phi_2$ for all $i\geq 0$, or there exists $i\ge 0$
  with $\xi_i\models\phi_1\wedge\phi_2$ and $\xi_j\vDash\phi_2$ for
  all $0\leq j< i$;
  \item $\xi \models X\ \phi$ iff $\xi_i\models \phi$.
\end{itemize}

According to the LTL semantics, it holds $\phi R\psi=\neg (\neg \phi U\neg
\phi)$. We use the usual abbreviations $\true=a \vee \neg a$,
$Fa=\mathsf{True}Ua$ and $Ga=\mathsf{False}Ra$.

\textbf{Notations.} Let $\phi$ be a formula written in
\emph{conjunctive form} $\phi = \bigwedge_{i\in I} \phi_i$ such that
the root operator of $\phi_i$ is not a conjunctive: then we define the
conjunctive formula set as $CF(\phi):=\{\phi_i \mid i\in I\}$.  When
$\phi$ does not include a conjunctive as a root operator,
$\mathit{CF}(\phi)$ only includes $\phi$ itself. For technical
reasons, we assume that $\mathit{CF}(\mathsf{True})=\emptyset$.  Our
construction requires that every atoms (properties) in the formula can
be varied from their positions. For example, for the formula $aUa$ -
we should consider the two of $a$s are identified syntactically
differently, similarly for the formula $aU\neg a$.

\subsection{Disjunctive Normal Form}\label{sec:dnf}
We introduce the notion of \textit{disjunctive-normal  form} for LTL formulas in the following.

\begin{definition}[disjunctive-normal form]\label{def:dnf}
  A formula $\phi$ is in \textit{disjunctive-normal form} (DNF) if
  it can be represented as $\phi:=\bigvee _i (\alpha_i\wedge X
  \phi_i)$, where $\alpha_i$ is a finite conjunction of literals,
  and $\phi_i = \bigwedge \phi_{i_j}$ where $\phi_{i_j}$ is
  either a literal, or an \textit{Until}, \textit{Next} or
  \textit{Release} formula.

We say $\alpha_i\wedge X \phi_i$ is
a \emph{clause} of $\phi$, and write $DNF(\phi)$ to denote all of the clauses.
\end{definition}

As seen in the introduction and motivating example, DNF form plays a
central role in our construction. Thus, we first discuss that any LTL
formula $\phi$ can be transformed into an equivalent formula in DNF
form.  The transformation is done in two steps: the first step is
according to the following rules:

\begin{lemma}\label{lemma:expansion}
\begin{enumerate}

\item $DNF(\alpha) =\{\alpha \wedge X(\mathsf{True})\}$ where $\alpha$ is a literal;
\item $DNF(X\phi) = \{\mathsf{True}\wedge X(\phi)\}$;
\item $DNF(\phi_1 U \phi_2) = DNF(\phi_2)\cup DNF( \phi_1 \wedge
  X(\phi_1 U \phi_2))$;
  \item $DNF(\phi_1 R \phi_2) = DNF(\phi_1 \wedge \phi_2) \cup DNF( \phi_2 \wedge X(\phi_1 R \phi_2))$;
  \item $DNF(\phi_1 \vee \phi_2) = DNF(\phi_1)\cup DNF(\phi_2)$;
  \item $DNF(\phi_1\wedge\phi_2) = \{(\alpha_1\wedge\alpha_2) \wedge X(\psi_1\wedge\psi_2)\mid \forall i=\mathit{1,2}. \ \alpha_i\wedge X(\psi_i)\in DNF(\phi_i)\}$;

\end{enumerate}
\end{lemma}

All of the rules above are self explained, following by the
definition of DNF, distributive and the expansion laws. What remains is how to deal with
the formulas in the \textit{Next} operator: by definition, in a
clause $\alpha_i\wedge X(\phi_i)$ the root operators in $\phi_i$
cannot be disjunctions. The equivalence $X (\phi_1 \vee
\phi_2) = X \phi_1 \vee X \phi_2$ can be applied repeatedly
to move the disjunctions out of the \textit{Next} operator. The
distributive law of disjunction over conjunctions allows us to bring
any formula into an equivalent DNF form:

\begin{theorem}\label{thm:transform}
  Any LTL formula $\phi$ can be transformed into an equivalent formula in
  disjunctive-normal form.
\end{theorem}

In our running example, we have $DNF(\phi_1)= DNF(\phi_2)=
DNF(\phi_3)= DNF(\phi_4)=\{c\wedge e\wedge X(\phi_1), b\wedge e\wedge
X(\phi_2), c\wedge d\wedge X(\phi_3), b\wedge d\wedge X(\phi_4)
\}$. Below we discuss the set of formulas that can be reached from a
given formula.

\begin{definition}[Formula Expansion]\label{def:expand}
  We write $\phi\tran{\alpha}\psi$ iff there exists $\alpha\wedge X(\psi) \in
  DNF(\phi)$. We say $\psi$ is expandable from $\phi$, written as
  $\phi\hookrightarrow\psi$, if there exists a finite expansion
  $\phi\tran{\alpha_1}\psi_1\tran{\alpha_2}\psi_2\tran{\alpha_3}\ldots
  \psi_n=\psi$. Let $EF(\phi)$ denote the set of all formulas that
  can be expanded from $\phi$. 
\end{definition}

The following theorem points out that $|EF(\lambda)|$ is bounded:
\begin{theorem}\label{thm:expand:bounded}
For any formula $\lambda$, $|EF(\lambda)|\leq 2^{n+1}$ where $n$ denotes the number of subformulas of $\lambda$.
\end{theorem}
\section{\emph{DNF-based}
B\"uchi Automaton Construction}\label{construction}
Our goal of this section is to construct the B\"uchi automaton
$\A_\lambda$ for $\lambda$. We establish a few simple properties of
general formulas that shall shed insights on the construction for the
$Release$-free ($Until$-free) formulas. We then define the labelled transition system for a formula. In the following three
subsections we present the construction for $Release$-free
($Until$-free) and general formulas, respectively.  

In the remaining of the paper, we fix $\lambda$ as the input LTL
formula.  All formulas being considered will vary over the set
$EF(\lambda)$, and $AP$ will denote the set of all literals appearing
in $\lambda$, and $\Sigma =2^{AP}$. 

\subsection{Transition Systems for LTL Formulas}
We first extend formula expansions
to subset in $\Sigma$:
\begin{definition}\label{def:expandset}
  For $\omega\in \Sigma$ and propositional formula $\alpha$,
  $\omega\models \alpha$ is defined in the standard way: if $\alpha$
  is a literal, $\omega\models \alpha$ iff $\alpha\in \omega$, and
  $\omega\models \alpha_1\wedge\alpha_2$ iff
  $\omega\models\alpha_1\wedge\omega\models\alpha_2$, and
  $\omega\models \alpha_1\vee\alpha_2$ iff
  $\omega\models\alpha_1\vee\omega\models\alpha_2$.

We write $\phi\tran{\omega}\psi$ if
$\phi\tran{\alpha}\psi$ and $w\models\alpha$.  For a word
$\eta=\omega_0\omega_1..\omega_k$, we write $\phi\tran{\eta}\psi$ iff
$\phi\tran{\omega_0}\psi_1\tran{\omega_1}\psi_2\tran{\omega_2}..\psi_{k+1}=\psi$.

For a run $\xi\in\Sigma^\omega$, we write $\phi\tran{\xi}\phi$
iff $\xi$ can be written as $\xi=\eta_0\eta_1\eta_2\ldots$ such that $\eta_i$ is a finite sequence, and  $\phi\tran{\eta_i}\phi$ for all $i\ge 0$.
\end{definition}

Below we provide a few interesting properties derived from our DNF normal forms.
\begin{lemma}\label{lemma:dnf}
Let  $\xi$ be a run and $\lambda$ a formula. Then,  for all $n\ge
1$,
$\xi\vDash\lambda\Leftrightarrow\lambda\tran{\xi^n}\phi\wedge\xi_n\vDash\phi$.
\end{lemma}
Essentially, $\xi\models\lambda$ is equivalent to that we can
reach a formula $\phi$ along the prefix $\xi^n$  such that the
suffix $\xi_n$ satisfies $\phi$. The following corollary is a direct
consequence of Lemma \ref{lemma:dnf} and  the fact that we have only
finitely many formulas in $EF(\lambda)$:
\begin{corollary}\label{coro:expand:existcycle}
  If $\xi\vDash\lambda$, then there exists $n\ge 1$ such that
  $\lambda\tran{\xi^n}\phi\wedge\xi_n\vDash\phi \wedge
  \phi\tran{\xi_n}\phi$.
On the other side, if $\lambda\tran{\xi^n}\phi\wedge\xi_n\vDash\phi \wedge \phi\tran{\xi_n}\phi$, then $\xi\vDash\lambda$.
\end{corollary}

This corollary gives the hint that after a finite prefix we can focus
on whether the suffix satisfies the \emph{looping formula} $\phi$,
i.e,. those $\phi$ with $\phi\hookrightarrow\phi$.
From Definition~\ref{def:expand} and the expansion rules for LTL
formulas, we have the following corollary:
\begin{corollary}\label{coro:expand:cycle}
  If $\lambda\hookrightarrow\lambda$ holds and $\lambda\neq\mathsf{True}$, then there is at least one \textit{Until} or \textit{Release} formula in $CF(\lambda)$.
\end{corollary}

As we described in previous, the elements in $EF(\lambda)$ and its
corresponding DNF-normal forms naturally induce a labelled transition
system, which can be defined as follows:
\begin{definition}[LTS for $\lambda$]
  The labelled transition system $TS_{\lambda}$ generated from the
  formula $\lambda$ is a tuple $\langle \Sigma,S, \delta, S_0\rangle$: where 
  $\Sigma = AP$, $S  = EF(\lambda)$, $S_0=\{\lambda\}$ and $\delta$ is defined as follows:
  $\psi\in\delta(\phi, \omega)$ iff $\phi\tran{\omega}\psi$ holds,
  where $\phi,\psi\in EF(\lambda)$ and $\omega\in \Sigma$.
\end{definition}

\subsection{B\"uchi automata for Release/Until-free Formulas}
The following lemma is a special instance of our central theorem
\ref{thm:central}.  It states properties of accepting runs with respect to Release/Until-free formulas:

\begin{lemma}\label{lem:releasefree}
  \begin{enumerate}
  \item   Assume $\lambda$ is  $Release$-free. Then, $\xi\vDash\lambda\Leftrightarrow\exists n\cdot\lambda\tran{\xi^n}\mathsf{True}$.
  \item Assume $\lambda$ is  $Until$-free. Then $\xi\vDash\lambda\Leftrightarrow\exists n,\phi\cdot\lambda\tran{\xi^n}\phi\wedge\phi\tran{\xi_n}\phi$.
  \end{enumerate}
\end{lemma}
Essentially, If $\lambda$ is Release-free, we will
reach $\true$ after finitely many steps; If $\lambda$ is Until-free we
will reach a looping formula after finitely many steps. The B\"uchi
automaton for Release-free or Until-free formulas will be directly
obtained by equipping the LTS with the set of accepting states:

\begin{definition}[$\A_\lambda$ for Release/Until-free formulas]\label{def:rfreeautomaton}
  For a Release/Until-free formula $\lambda$, we define the B\"uchi automaton
  $\mathcal{A}_\lambda=(S, \Sigma, \delta , S_0, F)$ where $TS_{\lambda}=\langle \Sigma, S, \delta, S_0\rangle$. The set $F$ is defined by:
  $F=\{\true\}$ if $\lambda$ is Release-free while $F=S$ if $\lambda$ is Until-free.
\end{definition}

Notably, $\true$ is the only accepting state for $\A_\lambda$ when
$\lambda$ is Release-free while all the states are accepting ones if
it is Until-free.

\begin{theorem}[Correctness and Complexity]\label{bound}
  Assume $\lambda$ is $Until$-free or $Release$-free.
Then, for any sequence
  $\xi\in \Sigma^\omega$, it holds $\xi\vDash \lambda$ iff $\xi$ is accepted
  by $\mathcal{A}_{\lambda}$. Moreover,
  $\mathcal{A}_{\lambda}$ has at most $2^{n+1}$ states,
  where $n$ is the number of subformulas in $\lambda$.
\end{theorem}
\begin{proof}
The proof of the correctness is trivial according to Lemma~\ref{lem:releasefree}: 1) if $\lambda$ is Release-free, then every run $\xi$ of $\mathcal{A}_{\lambda}$ can run across the $\true$-state\footnote{In this paper we use $\phi$-state to denote the state representing the formula $\phi$.} infinitely often iff it satisfies $\exists n\geq 0\cdot\lambda\tran{\xi^n}\true$, that is, $\xi\vDash\lambda$; 2) if $\lambda$ is Until-free, then $\xi\vDash\lambda$ iff $\exists n, \phi\cdot\lambda\tran{\xi^n}\phi\wedge\phi\tran{\xi_n}\phi$, which will run across $\phi$-state infinitely often so that is accepted by $\mathcal{A}_{\lambda}$ according to the construction.

The upper bound is a direct consequence of Theorem \ref{thm:expand:bounded}.
\end{proof}

\subsection{Central Theorem for General Formulas}
In the previous section we have constructed B\"uchi automaton for
Release-free or Until-free formulas, which is obtained by equipping
the defined LTS with appropriate accepting states.  For general
formulas, this is however slightly involved. For instance, consider the
LTS of the formula $\phi=G(bUc\wedge dUe)$ in our running example:
there are infinitely many runs starting from the initial state $s_1$,
but which of them should be accepting?  Indeed, it is not obvious how
to identify the set of accepting states. In this section we present
our central theorem for general formulas aiming at identifying the
accepting runs.

\begin{figure}
\centering
  \includegraphics[scale = 0.5]{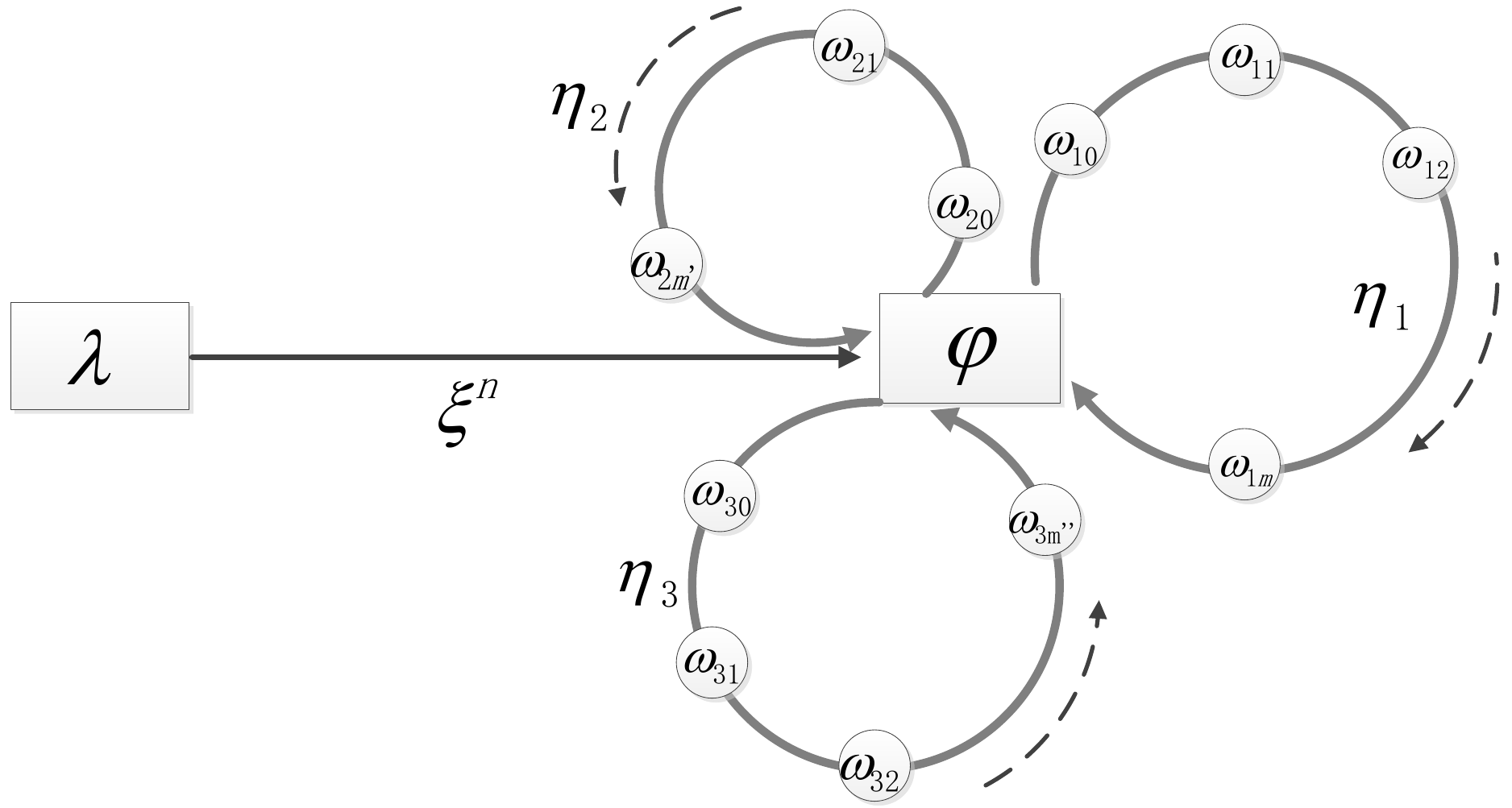}\\
  \caption{A snapshot illustrating the relation $\xi\models\lambda$}\label{fig:central_theorem}
\end{figure}

Assume the run $\xi=\omega_0\omega_1\ldots$ satisfies the formula
$\lambda$.  We refer to
$\lambda(=\phi_0)\tran{\alpha_0}\phi_1\tran{\alpha_1}\phi_2\ldots$ as
an expansion path from $\lambda$, which corresponds to a path in the
LTS $TS_\lambda$, but labelled with propositional formulas. Obviously, $\xi\models\lambda$ implies that there
exists an expansion path in $TS_\lambda$ such that
$\omega_i\models\alpha_i$ for all $i\ge 0$.  As the set $EF(\lambda)$
is finite, we can find a looping formula $\phi=\phi_i$ that occurs
\emph{infinitely often} along this expansion path. On the other side,
we can \emph{partition} the run $\xi$ into sequences
$\xi=\eta_0\eta_1\ldots$ such each finite sequence $\eta_i$ is
consistent with respect to one loop $\phi \hookrightarrow \phi$ along
the expansion path. This is illustrated in Figure
\ref{fig:central_theorem}. The definition below formalizes the notion
of consistency for finite sequence:

\begin{definition}\label{def:finitestepsat}
  Let $\eta=\omega_0\omega_1\ldots\omega_n$ ($n\geq 0$) be a finite sequence. Then, we say
that $\eta$ satisfies the LTL formula $\phi$, denoted
  by $\eta\models _f \phi$, if the following conditions are satisfied:
   \begin{itemize}
   \item there exists
     $\phi_0=\phi\tran{\alpha_0}\phi_1\tran{\alpha_1}\ldots\tran{\alpha_{n}}\phi_{n+1}
     =\psi$ such that $\omega_i\models\alpha_i$ for $0\le i \le n$,
     and with $S:=\bigcup_{0\leq j\leq n} CF(\alpha_j)$, it holds
        \begin{enumerate}
           \item if $\phi$ is a literal then $\phi\in S$ holds;
           \item if $\phi$ is $\phi_1U\phi_2$ or $\phi_1R\phi_2$ then $S\models_f\phi_2$ holds;
           \item if $\phi$ is $\phi_1\wedge \phi_2$ then $S\models_f\phi_1 \wedge S\models_f\phi_2$ holds;
           \item if $\phi$ is $\phi_1\vee \phi_2$ then $S\models_f\phi_1 \vee S\models_f\phi_2$ holds;
           \item if $\phi$ is $X \phi_2$ then $S\models_f\phi_2$ holds;
         \end{enumerate}
  \end{itemize}
\end{definition}

This predicate specifies whether the given finite sequence $\eta$ is
consistent with respect to the finite expansion
$\phi_0=\phi\tran{\alpha_0}\phi_1\tran{\alpha_1}\ldots\tran{\alpha_{n}}\phi_{n+1}
=\psi$. The condition $\omega_i\models\alpha_i$ requires that the
finite sequence $\eta$ is consistent with respect to the labels along
the finite expansion from $\phi_0$. The rules for literals and Boolean
connections are intuitive. For Until operator $\phi_1 U \phi_2$, it is
defined recursively by $S\models_f\phi_2$: as to make the Until
subformula being satisfied, we should make sure that $\phi_2$ holds
under $S$. Similar, for release operator $\phi_1 R\phi_2$, we know
that $\phi_1\wedge \phi_2$ or $\phi_2$ plays a key role in an
accepting run of $\phi_1 R \phi_2$. Because $\phi_1\wedge \phi_2$
implies $\phi_2$, and with the rule (4) in the definition, we have
$S\models_f\phi_1 R\phi_2 \equiv S\models_f\phi_2$. Assume
$\phi=X\phi_2$. As $\mathit{CF}(\mathsf{True})$ is defined as
$\emptyset$, we have $\eta\models_f \phi$ iff $\eta'\models_f\phi_2$
with $\eta'=\omega_1\omega_2\ldots\omega_n$.

The predicate $\models_f$ characterizes whether the prefix of an
accepting run contributes to the satisfiability of $\lambda$. The idea
comes from Corollary~\ref{coro:expand:existcycle}: Once $\phi$ is
expanded from itself infinitely by a run $\xi$ as well as
$\xi\models\phi$, there must be some common feature each time $\phi$
loops back to itself. This common feature is what we defined in
$\models_f$. In our running example, consider the finite sequence
$\eta=\{b,d\}\{b,d\}\{c,e\}$ corresponding to the path $s_1s_4s_4s_1$:
according to the definition $\eta\models_f\phi_1$ holds. For
$\eta=\{b,d\}\{b,d\}\{b,d\}$, however, $\eta\not\models_f\phi_1$.

With the notation $\models_f$, we study below properties for the
looping formulas, that will lead to our \textit{central theorem}.



\begin{lemma}[Soundness]\label{lemma:finitesat:infsat}
  Given a looping formula $\phi$ and an infinite word $\xi$, let $\xi
  =\eta_1\eta_2\ldots$. If $\forall i\geq 1\cdot
  \phi\tran{\eta_i}\phi\wedge \eta_i\models_f\phi$, then $\xi\vDash
  \phi$.
\end{lemma}

The soundness property of the looping formula says that if there
exists a partitioning $\xi = \eta_1\eta_2...$ such that $\phi$ expends
to itself by each $\eta_i$ and $\eta_i\models_f\phi$ holds, then
$\xi\models\phi$.

\begin{lemma}[Completeness]\label{lemma:completeness}
  Given a looping formula $\phi$ and an infinite word $\xi$, if $\phi\tran{\xi}\phi$ and $\xi\vDash\phi$ holds,
  then there exists a partitioning $\eta_1\eta_2\ldots$ for $\xi$, i.e.
  $\xi=\eta_1\eta_2\ldots$, such that for all $i\geq 0$, $\phi\tran{\eta_i}\phi\wedge \eta_i\models_f\phi$ holds.
\end{lemma}

The completeness property of the looping formula states the other
direction. If $\phi\tran{\xi}\phi$ as well as $\xi\models\phi$, we can find a
partitioning $\eta_1\eta_2\ldots$ that  makes $\phi$ expending to
itself by each $\eta_i$ and $\eta_i\models_f\phi$ holds.
Combining Lemma 6, Lemma 7 and Corollary 1, we have our central
theorem:

\begin{theorem}[Central Theorem]\label{thm:central}
   Given a formula $\lambda$ and an infinite word $\xi$,  we have
   \begin{align*}
         \xi\vDash \lambda \Leftrightarrow \exists \phi, n\cdot \lambda\tran{\xi^n}\phi\wedge \exists \xi_n = \eta_1\eta_2\ldots \cdot \forall i\geq 1\cdot \phi\tran{\eta_i}\phi\wedge \eta_i\models_f\phi
   \end{align*}

\end{theorem}

The central theorem states that given a formula $\lambda$, we 
can always extend it to a looping formula which satisfies the 
soundness and completeness properties. Reconsider
Figure~\ref{fig:central_theorem}: formula $\lambda$ extends to the
looping formula $\phi$ by $\xi^n$, and $\xi_n$ can be partitioned into sequences $\eta_1\eta_2\ldots$. The loops from $\phi$  correspond to these finite sequences $\eta_i$ in the sense $\eta_i\models_f\phi$.

\subsection{B\"uchi automata for General Formulas}\label{sec:automationgeneration}
Our central theorem sheds insights about the correspondence between
the accepting run and the expansion path from $\lambda$. However, how
can we guarantee the predicate $\models_f$ for looping formulas in the theorem?  We need
the last ingredient for starting our automaton construction: we
extract the \emph{obligation sets} from LTL formulas that will enable
us to characterize $\models_f$.

\begin{definition}\label{def:obligationset}
  Given a formula $\phi$, we define its obligation set, i.e. $OS_{\phi}$, as follows:
  \begin{enumerate}
    \item If $\phi=p$, $OS_{\phi}=\{ \{p\}\}$;
    \item If $\phi=X\psi$, $OS_{\phi}=OS_{\psi}$;
    \item If $\phi=\psi_1\vee\psi_2$, $OS_{\phi}=OS_{\psi_1}\cup OS_{\psi_2}$;
    \item If $\phi=\psi_1\wedge\psi_2$, $OS_{\phi}=\{S_1\cup S_2 \mid S_1\in OS_{\psi_1}\wedge S_2\in OS_{\psi_2}\}$;
    \item If $\phi=\psi_1 U\psi_2$ or $\psi_1 R\psi_2$, $OS_{\phi}=OS_{\psi_2}$;
  \end{enumerate}
  For every element set $O\in OS_{\phi}$, we call it the obligation of $\phi$.
\end{definition}

The obligation set provides all obligations (elements in obligation
set) the given formula is supposed to have. Intuitively, a run $\xi$
accepts a formula $\phi$ if $\xi$ can eliminate the obligations of
$\phi$. Take the example of $G(a R b)$, the run $(b)^{\omega}$
accepts $a R b$, and the run eliminates the obligation set $\{\{b\}\}$
infinitely often.

Notice the similarity of the definition of the obligation set and the
predicate $\models_f$.  For instance, the obligation set of ${\phi_1 R
  \phi_2}$ is the obligation set of $\phi_2$, which is similar in the
definition of $\models_f$.  The interesting rule is the conjunctive
one. For obligation set $OS_{\phi}$, there may be more than one
element in $OS_{\phi}$.  However, from the view of satisfiability, if
one obligation in $OS_{\phi}$ is satisfied, we can say the obligations
of $\phi$ is fulfilled. This view leads to the definition of the
conjunctive rule. For $\psi_1\wedge\psi_2$, we need to fulfill the
obligations from both $\psi_1$ and $\psi_2$, which means we have to
trace all possible unions from the elements of $OS_{\psi_1}$ and
$OS_{\psi_2}$. For instance, the obligation set of $G(a U b \wedge c U
(d\vee e))$ is $\{\{b, d\}, \{b, e\}\}$.
The following lemmas gives the
relationship of $\models_f$ and \textit{obligation set}.

\begin{lemma}\label{lemma:obligaionandsatonce}
  For all $O\in OS_{\phi}$, it holds $O\models_f\phi$. On the other side,
$S\models_f\phi$ implies that $\exists O\in OS_{\phi}\cdot O\subseteq S$.
\end{lemma}

For our input formula $\lambda$, now we discuss how to construct the
B\"uchi automaton $\A_\lambda$.  We first describe the
states of the automaton. A state will be consisting of the formula
$\phi$ and a
\emph{process set} that keeps track of properties have been satisfied so far. Formally:

\begin{definition}[states of the automaton for $\lambda$]
  A state is a tuple $\langle\phi,P\rangle$ where $\phi$ is a formula from $EF(\lambda)$, and $P\subseteq AP$
  is a \textit{process set}.
\end{definition}

Refer again to Figure \ref{fig:central_theorem}: reading the input finite
sequence $\eta_1$, each element in the process set $P_i$ corresponds
to a property set belonging to $AP$, which will be used to keep track whether all
elements in an obligation are met upon returning back to a $\phi$-state. If we
have $P_i=\emptyset$, we have successfully returned to the accepting
states. Now we have all ingredients for constructing our B\"uchi
automaton $\A_\lambda$:

\begin{definition}[B\"uchi Automaton $\A_\lambda$]
  The B\"uchi automaton for the formula $\lambda$ is defined as
  $\mathcal{A}_{\lambda}=(\Sigma, S, \delta, S_0, \mathcal{F})$, where
  $\Sigma = 2^{AP}$ and:
  \begin{itemize}
  \item $S=\{\langle\phi,P\rangle\mid \phi\in EF(\lambda)\}$ is the set of states;
  \item $S_0=\{\langle \lambda,\emptyset\rangle\}$ is the set of initial states;
  \item $\mathcal{F} =\{\langle\phi,\emptyset\rangle\mid\phi\in EF(\lambda)\}$ is the set of accepting states;
  \item Let states $s_1,s_2$ with $s_1=\langle\phi_1,P_1\rangle$, $s_2=\langle\phi_2,P_2\rangle$ and
    $w\subseteq 2^{AP}$. Then, $s_2\in \delta(s_1, \omega)$ iff there exists
    $\phi_1\tran{\alpha}\phi_2$ with $\omega\models\alpha$ such that the corresponding $P_2$ is
    updated by:
    \begin{enumerate}
    \item $P_2 = \emptyset$ if $\exists O\in OS_{\phi_2}\cdot O\subseteq P_1\cup CF(\alpha)$,
  \item $P_2 = P_1 \cup CF(\alpha)$ otherwise.
    \end{enumerate}
\end{itemize}
\end{definition}

The transition is determined by the expansion relation
$\phi_1\tran{\alpha}\phi_2$ such that $\omega\models\alpha$. The process
set $P_2$ is updated by $P_1\cup CF(\alpha)$ unless there is no element set $O\in OS_{\phi_2}$ such that $P_1\cup CF(\alpha)\supseteq O$.
In that case $P_2$ will be set to $\emptyset$ and the corresponding state will be recognized as an accepting one.

Now we state the correctness of our construction:

\begin{theorem}[Correctness of Automata Generation]\label{thm:correct}
  Let $\lambda$ be the input formula. Then, for any sequence
  $\xi\in \Sigma^\omega$, it holds $\xi\vDash \lambda$ iff $\xi$ is accepted
  by $\mathcal{A}_{\lambda}$.
\end{theorem}
The correctness follows mainly from the fact that our construction
strictly adheres to our central theorem (Theorem
\ref{thm:central}).

We note that two very simple optimizations can be identified for our
construction:
\begin{itemize}
\item If two states have the same DNF normal form and the same
process set $P$, they are identical. Precisely, we merge
states $s_1=\langle \phi_1, P_1\rangle$ and $s_2=\langle \phi_2,
P_2\rangle$ if $DNF(\phi_1)=DNF(\phi_2)$, and $P_1=P_2$;
\item The elements in the process set $P$ can be restricted into those
  atomic propositions appearing in $OS_{\phi}$: Recall here $\phi\in
  EF(\lambda)$. One can observe directly that only those properties
  are used for checking the \textit{obligation} conditions, while
  others will not be used so that it can be omitted in the process
  set $P$.
\end{itemize}
Now we can finally explain a final detail of our running example:

\begin{example}
  In our running example state $s_1$ is the accepting state of the
  automaton. It should be mentioned that the state $s_2$ =
  $\langle\phi_2,\{e\}\rangle$
  originally has an edge labeling $c\wedge d$ to the state
  $\langle\phi_3,\emptyset\rangle$ according to our construction, which is a
  new state.  However, this state is equivalent with $s_1=\langle
  \phi_1,\emptyset\rangle$, as $\phi_1$
  and $\phi_3$ have the same DNF normal form. So these two states are
  merged. The same cases occur on state $s_3$ to state $s_1$ with the
  edge labeling $b\wedge e$, state $s_2$ to state $s_2$ with the edge
  labeling $b\wedge d$ and etc. After merging these states, we have
  the automaton as depicted in Figure \ref{fig:buechi}.
\end{example}

\begin{theorem}[Complexity]\label{generalbound}
  Let $\lambda$ be the input formula. Then the B\"uchi automaton
  $\mathcal{A}_{\lambda}$ has the upper bound $2^{2n+1}$, where $n$ is
  the number of subformulas in $\lambda$.
\end{theorem}

The number of states is bounded by $2^{n+1}\cdot 2^{|AP|} \le
2^{2n+1}$. Recall in the construction $AP$ is the set of atomic
prepositions appearing in $\lambda$, thus $|AP|$ is much smaller than
$n$ in general.  We remark that the first part $2^{n+1}$ is much
smaller in practice due to equivalent DNF representations. Indeed, it
can be reduced to $2^{dnf(\lambda)+1}$ where $dnf(\lambda)$ denotes the number of equivalence   classes of $EF(\lambda)$ induced by
equivalent DNF representations. 
In our running example, all of the
formulas have the same DNF normal form, thus this part is equal to
$2^{1+1}=4$. On the other side, the second part $2^{|AP|}$ can be
further reduced to the set of atomic propositions that appear in the
obligation sets: in our running example this is $|\{c,e\}|$.
\section{Discussion}\label{discussion}
In this section, we discuss the relationship and differences between
our proposed approach and the tableau construction.

Generally speaking, our approach is essentially a tableau one that
is based on the expansion laws of $Until$ and $Release$ operators. The
interesting aspect of our approach is the finding of a special
normal formal with its DNF-based labeled transition system, which is closely
related to the B\"uchi automaton under construction. The tableau
approach explicitly expands the formula recursively based on the
semantics of LTL formulas while the nodes of the potential automaton
are split until no new node is generated. However, our approach
first studies the LTL normal forms to discover the obligations we
have to fulfill for the automaton to be generated, and then presents
a simple mapping between LTL formulas into B\"uchi automata.

The insight behind our approach is adopting a different view on the
accepting conditions.  The tableau approach focuses on the
$Until$-operator. For instance, to decide the accepting states, the
tableau approach needs to trace all the $Until$-subformulas and records
the ``eventuality'' of $\psi$ in $\phi U \psi$, which leads to the
introduction of the \textit{Generalized B\"{u}chi Automata} (GBA) in
tableau approach. However, our approach focuses on the
\textit{looping formulas}, which potentially consist of the
accepting states. Intuitively, an infinite sequence (word) will
satisfy the formula $\lambda$ iff $\lambda$ can expand to some
looping formula $\phi$ which can be satisfied by the suffix of the word
removing the finite sequence arriving at $\phi$. 
The key point of our approach is to introduce the static obligation
set for each formula in the DNF-based labeled transition system, which
indicates that an accepting run is supposed to infinitely fulfil one
of the obligations in the obligation set.  Thus, the obligation set
gives the "invariability" for general formulas instead of the
``eventuality'' for $Until$-formulas. 
In the approach,  we use a process set to record the obligation that
formula $\phi$ has been satisfied from its last appearance. Then, we
would decide the accepting states easily when the process set
fulfills one obligation in the  obligation set of $\phi$ (We reset
it empty afterwards). One can also note our approach is  on-the-fly:
the successors of the current state can be obtained as soon as its
DNF normal form is acquired.

The most interesting part is that, our approach can give a more precise theoretical upper bound
for the complexity of the translation when comparing to the tableau framework (Theorem \ref{generalbound}). And a
better one can be acquired when the formulas are restricted into Release-free or Until-free (Theorem \ref{bound}).

\section{Conclusion}\label{conclusion}

In this paper, we propose the \textit{disjunctive-normal forms}  for  LTL
formulas. Based on the DNF representation, we introduce the
DNF-based labeled transition system for formula $\lambda$ and study the
relationship between the transition system and the B\"uchi automata
for $\lambda$. Thus, a simple but on-the-fly automata construction is
achieved. When the formula under construction is
Release/Until-free, our construction is very straightforward in
theory, and leads to at most $2^{n+1}$ states. In the general
way, our approach gives a more precise bound of $2^{2n+1}$ compared
to the one of $2^{O(n)}$ for tableau construction.

\newpage
\appendix
\section{Proofs}
\subsection{Proof of Theorem \ref{thm:transform}}
Let $\phi$ be a formula $\varphi =
\bigvee_{i\in I} \varphi_i$ such that the root operator of
$\varphi_i$ is not a disjunctive: then we define the disjunctive
formula set as $DF(\varphi):=\{\varphi_i \mid i\in I\}$.  When
$\varphi$ does not include a disjunctive as a root operator,
$DF(\varphi)$ only include $\varphi$ itself.
\begin{proof}
  We first can directly use the rules in Lemma~\ref{lemma:expansion} to generate an intermediate normal form for $\phi$, whose format is $\bigvee _i (\alpha_i\wedge X \varphi_i)$ where $\alpha_i$ is an propositional formula and $\varphi_i$ is an LTL formula without any constraint in Definition~\ref{def:dnf}. We denote the set of this intermediate normal form of the formula $\phi$ as $DNF_1(\phi)$;

  Second we prove any intermediate normal form can be changed to the \\\textit{disjunctive-normal form}. Intuitively, one can easily find for each $\alpha_i$ and $\varphi_i$ the corresponding $DF(\alpha_i)$ and $DF(\varphi_i)$ can be obtained trivially. Then we can get the final \emph{disjunctive-normal form} through the following two steps:
  \begin{enumerate}
    \item $DNF_2(\phi)=\{\alpha_i\wedge X\psi\mid \alpha\wedge X\psi\in DNF_1(\phi)\wedge\alpha_i\in DF(\alpha)\}$;
    \item $DNF(\phi)=\{\alpha\wedge X\psi_i\mid\alpha\wedge X\psi\in DNF_2(\phi)\wedge\psi_i\in DF(\psi)\}$.
  \end{enumerate}
\end{proof}

\subsection{Proof of Theorem \ref{thm:expand:bounded}}
Let $n$ be the number of subformulas in $\lambda$. Moreover, let
$cl(\lambda)$ be the set of subformulas in $\lambda$ and
$\mathsf{True}$, so obviously $|cl(\lambda)|=n+1$. Before the proof we
introduce two lemmas first.
\begin{lemma}\label{lemma:dnf:finite}
  Let $\alpha\wedge X\psi\in DNF(\varphi)$, then $CF(\psi)\subseteq cl(\varphi)$.
\end{lemma}
\begin{proof}
  We prove it by structural induction over $\phi$.
  \begin{itemize}
    \item Basic step: If $\varphi$ is the case of the literal $p$, then since $p = p\wedge X\mathsf{True}$, so obviously $CF(\mathsf{True})\subseteq cl(\varphi)$.
    \item Inductive step: If the formulas $\varphi_i$ ($i=\mathit{1,2}$) satisfy $\alpha\wedge X\psi\in DNF(\varphi_i) \Rightarrow CF(\psi)\subseteq cl(\varphi_i)$, then:
        \begin{enumerate}
          \item If $\phi=\phi_1\vee\phi_2$, we know $cl(\phi)=cl(\phi_1)\cup cl(\phi_2)\cup \{\phi_1\vee\phi_2\}$. According to Lemma~\ref{lemma:expansion}.5 we have $\alpha\wedge X\psi\in DNF(\varphi)\Rightarrow\alpha\wedge X\psi\in DNF(\varphi_1)\cup DNF(\varphi_2)$, then by induction hypothesis we have $CF(\psi)\subseteq cl(\phi_1)\cup cl(\phi_2)$, so $CF(\psi)\subseteq cl(\phi)$;
          \item If $\phi=X\phi_1$, we know $cl(\phi)=cl(\phi_1)\cup \{X\phi_1\}$. According to Lemma~\ref{lemma:expansion}.2 we have $\alpha\wedge X\psi\in DNF(\varphi)\Rightarrow\psi =\phi_1$, so $CF(\psi)\subseteq cl(\phi_1)\subseteq cl(\phi)$;
          \item If $\phi=\phi_1\wedge\phi_2$, we know $cl(\varphi_1\wedge\varphi_2)=\{\varphi_1\wedge\varphi_2\}\cup cl(\varphi_1)\cup cl(\varphi_2)$. According to Lemma~\ref{lemma:expansion}.6 we know $\alpha\wedge X\psi\in DNF(\varphi_1\wedge\varphi_2)\Rightarrow\exists\alpha_1\wedge X\psi_1\in DNF(\phi_1), \alpha_2\wedge X\psi_2\in DNF(\phi_2)\cdot\alpha = \alpha_1\wedge\alpha_2\wedge\psi = \psi_1\wedge\psi_2$. Then by induction hypothesis we have $CF(\psi_1)\subseteq cl(\phi_1)$ and $CF(\psi_2)\subseteq cl(\phi_2)$, so $CF(\psi)\subseteq cl(\varphi_1)\cup cl(\varphi_2)\subseteq cl(\varphi_1\wedge\varphi_2)$;
          \item If $\phi=\phi_1 U\phi_2$, we know $cl(\varphi_1 U\varphi_2)=cl(\varphi_1)\cup cl(\varphi_2)\cup \{\varphi_1 U \varphi_2\}$. According to Lemma~\ref{lemma:expansion}.3 if $\alpha\wedge X\psi\in DNF(\varphi_2)$ then $CF(\psi)\subseteq cl(\varphi_2)$ directly by induction hypothesis, else if $\alpha\wedge X\psi\in \{\alpha\wedge X(\psi_1\wedge \varphi_1 U \varphi_2)\mid \alpha\wedge X\psi_1\in DNF(\varphi_1)\}$ then by induction hypothesis we have $CF(\psi)=CF(\psi_1)\cup\{\varphi_1 U \varphi_2\}\subseteq cl(\varphi_1)\cup \{\varphi_1 U \varphi_2\}\subseteq cl(\varphi_1 U \varphi_2)$;
          \item If $\phi=\phi_1 R \phi_2$ one can also prove in the similar way that $\alpha\wedge X\psi\in DNF(\varphi)\Rightarrow CF(\psi)\subseteq cl(\varphi)$.
        \end{enumerate}
  \end{itemize}
\end{proof}

\begin{lemma}\label{lemma:expand:finite}
  Let $\psi \in EF(\varphi)$ then $CF(\psi)\subseteq cl(\varphi)$;
\end{lemma}
\begin{proof}
  We prove it by induction over the number of steps that $\psi$ can be reached from $\phi$.

  \begin{itemize}
  \item   Base step: If $\alpha\wedge X\psi\in DNF(\varphi)$ then according to Lemma~\ref{lemma:dnf:finite} we know $CF(\psi)\subseteq cl(\varphi)$.
  \item Induction step: If $\exists\varphi\rightarrow\psi_1\rightarrow\psi_2\rightarrow\ldots\psi_k=\psi$ where $k\geq 1$ and $CF(\psi)\subseteq cl(\varphi)$ hold, then according to Lemma~\ref{lemma:dnf:finite} we know for all $\nu\in CF(\psi)$ we have $\beta\wedge X\mu\in DNF(\nu)\Rightarrow CF(\mu)\subseteq cl(\nu\}\subseteq cl(\varphi)$. Then according to Lemma~\ref{lemma:expansion}.6 we know $\forall \alpha\wedge X\psi' \in DNF(\psi)\cdot CF(\psi')\subseteq cl(\varphi)$ holds. That is,
      if $\psi$ can be reached from $\phi$ in $k$ steps and $CF(\psi)\subseteq cl(\phi)$ holds, then any $\psi'$ can be reached from $\phi$ in $k+1$ steps also has $CF(\psi')\subseteq cl(\phi)$.
  \end{itemize}
\end{proof}

Now come to prove Theorem~\ref{thm:expand:bounded}. From Lemma~\ref{lemma:expand:finite} we know for all $\psi\in EF(\lambda)$ if $\mu\in CF(\psi)$ then we have $\mu\in cl(\lambda)$.
So the elements number in $CF(\psi)$ can not exceed the number of $cl(\lambda)$, i.e.
$|CF(\psi)|\leq |cl(\lambda)|$. Thus $|EF(\lambda)|\leq 2^{|cl(\lambda)|}=2^{n+1}$.

\subsection{Proof of Lemma \ref{lemma:obligaionandsatonce}}
We first prove the first part of the lemma by induction over the formula $\phi$.
  \begin{itemize}
   \item Basic step: If $\phi = p$, then $OS_{\phi} = \{\{p\}\}$, and $\{p\},\models_f p$ obviously true.
   \item Inductive step: If for the formulas $\psi_i$ ($i = 1,2$), $\forall O\in OS_{\psi_i}\cdot O\models_f\psi_i$ holds. Then we have:
   \begin{enumerate}
      \item If $\phi = X \psi_1$, then $OS_{\phi} = OS_{\psi_1}$. Since for each $O$ in $OS_{\phi}$, the predicate $O\models_f\phi\equiv O\models_f \psi_1$ according to its definition, and since $OS_{\phi} = OS_{\psi_1}$ so $O\in OS_{\psi_1}$. Then by induction hypothesis we know $O\models_f\psi_1$ holds thus $O\models_f\phi$ holds.
      \item If $\phi = \psi_1\vee\psi_2$, then $OS_{\phi} = OS_{\psi_1}\cup OS_{\psi_2}$, so we know $\forall O\in OS_{\phi}\cdot O\in OS_{\psi_1}\vee O\in OS_{\psi_2}$. Then since $O\models_f\phi \equiv O\models_f\psi_1\vee O\models_f\psi_2$, and by induction hypothesis $O\models_f \psi_1$ holds when $O\in OS_{\psi_1}$ while $O\models_f\psi_2$ holds when $O\in OS_{\psi_2}$. Due to $O\in OS_{\psi_1}\vee O\in OS_{\psi_2}$ so $O\models_f\phi \equiv O\models_f \psi_1\vee O\models_f \psi_2$ is true.
      \item If $\phi = \psi_1\wedge\psi_2$, then $OS_{\phi} = \{S_1\cup S_2\mid S_1\in OS_{\psi_1}\wedge S_2\in OS_{\psi_2}\}$.
          Then $\forall O\in OS_{\phi}\exists S_1\in OS_{\psi_1}, S_2\in OS_{\psi_2}\cdot O = S_1\cup S_2$. By induction hypothesis that $S_1\models_f \psi_1$ and $S_2\models_f \psi_2$ are true, thus $O\models_f\phi\equiv S_1\cup S_2\models_f \psi_1\wedge S_1\cup S_2\models_f\psi_2$ holds.
      \item If $\phi = \psi_1 U\psi_2$, then $OS_{\phi} = OS_{\psi_2}$. Since for each $O$ in $OS_{\phi}$ $O\models_f\phi\equiv O\models_f\psi_2$, and by induction hypothesis $O\models_f\psi_2$ holds, so $O\models_f\phi$ also holds. Similarly one can prove the situation when $\phi = \psi_1 R\psi_2$ and we omit it here.
   \end{enumerate}
  \end{itemize}

  We then prove the second part of the lemma also by induction over the formula $\phi$.
  \begin{itemize}
   \item Basic step: If $\phi = p$, then $OS_{\phi} = \{\{p\}\}$, and $S\models_f p\Rightarrow p\in S$. So obviously $\exists O\in OS_{\phi}\cdot O\subseteq S$.
   \item Inductive step: If for the formulas $\psi_i$ ($i = 1,2$), $S\models_f\psi_i\Rightarrow\exists O_i\in OS_{\psi_i}\cdot O\subseteq S$ holds. Then we have:
   \begin{enumerate}
      \item If $\phi = X \psi_1$, then we know $OS_{\phi} = OS_{\psi_1}$ and $S\models_f \phi\equiv S\models_f\psi_1$. Since by induction hypothesis $S\models_f\psi_1\Rightarrow\exists O\in OS_{\psi_1}\cdot O\subseteq S$, and $OS_{\psi_1}=OS_{\phi}$, so $O \in OS_{\phi}$. Thus $S\models_f\phi\Rightarrow\exists O\in OS_{\phi}\cdot O\subseteq S$ holds.
      \item If $\phi = \psi_1\vee\psi_2$, then we have $OS_{\phi} = OS_{\psi_1}\cup OS_{\psi_2}$ and $S\models_f\phi\equiv S\models_f\psi_1\vee S\models_f\psi_2$. By induction hypothesis $S\models_f\psi_1\Rightarrow\exists O\in OS_{\psi_1}\cdot O\subseteq S$ and $S\models_f\psi_2\Rightarrow\exists O\in OS_{\psi_2}\cdot O\subseteq S$, so $S\models_f\phi\Rightarrow\exists O\in OS_{\psi_1}\cup OS_{\psi_2}\cdot O\subseteq S$, in which $OS_{\psi_1}\cup OS_{\psi_2}$ is exactly $OS_{\phi}$. Thus $S\models_f\phi\Rightarrow\exists O\in OS_{\phi}\cdot O\subseteq S$ holds.
      \item If $\phi = \psi_1\wedge\psi_2$, then $OS_{\phi} = \{S_1\cup S_2\mid S_1\in OS_{\psi_1}\wedge S_2\in OS_{\psi_2}\}$.
          Since $S\models_f \phi\equiv S\models_f \psi_1\wedge S\models_f\psi_2$, and by induction hypothesis we have $S\models_f\psi_i\Rightarrow\exists O_i\in OS_{\psi_i}\cdot O_i\subseteq S$, where $i = 1, 2$, so $S\models_f\phi\Rightarrow\exists O = O_1\cup O_2\cdot O\subseteq S$. Obviously $O\in OS_{\phi}$, so $S\models_f\phi\Rightarrow\exists O\in OS_{\phi}\cdot O\subseteq S$ holds.
      \item If $\phi = \psi_1 U\psi_2$, then we know $OS_{\phi} = OS_{\psi_2}$ and $O\models_f \phi\equiv O\models_f\psi_2$. By induction hypothesis $O\models_f \psi_2\Rightarrow\exists O\in OS_{\psi_2}\cdot O\subseteq S$, and since $OS_{\phi}= OS_{\psi_2}$ so $O$ is also in $OS_{\phi}$. Thus $S\models_f\phi\Rightarrow\exists O\in OS_{\phi}\cdot O\subseteq S$ holds. Similarly one can prove the case when $\phi = \psi_1 R\psi_2$ and we omit it here.
   \end{enumerate}
  \end{itemize}

\subsection{Proof of Lemma \ref{lemma:finitesat:infsat}}
There are some other lemmas need to be introduced before we prove this lemma.

\begin{lemma}\label{lemma:subformula:equiv}
  $\mu\in cl(\nu)\wedge\nu\in cl(\mu)\Leftrightarrow\mu=\nu$.
\end{lemma}
\begin{proof}
  According to the definition of $cl$, it is definitely true.
\end{proof}

\begin{lemma}\label{lemma:formulacycle:exist}
  $\phi\hookrightarrow\phi \Rightarrow\exists\mu\in CF(\phi)\cdot cl(\mu)\cap CF(\phi) = \{\mu\}$.
\end{lemma}
\begin{proof}
  For each $\mu$ in $CF(\phi)$ let $S_{\mu} = cl(\mu)\cap CF(\phi)$, then we know easily $S_\mu\supseteq \{\mu\}$. If $\forall\mu\in CF(\phi)\cdot S_{\mu}\supset\{\mu\}$, then we know $\exists\mu_1\in S_{\mu}$ and $\mu_1\neq \mu$. Since $\mu_1$ is also in $CF(\phi)$, then according to the assumption $\exists\mu_2\in S_{\mu_1}$ and $\mu_2\neq \mu_1$. Moreover according to Lemma~\ref{lemma:subformula:equiv} $\mu_2\neq\mu$ also holds. However for $\mu_2$ it is also in $CF(\phi)$ and has at least one subformula $\mu_3$ in $CF(\phi)$ and $\mu_3\neq\mu_2$... Infinitely using this will cause $CF(\phi)$ be an infinite set - that is obviously impossible. So this lemma is true.
\end{proof}

\begin{lemma}\label{lemma:formulacycle:forall}
  $\phi\hookrightarrow\phi\Rightarrow\forall\phi\tran{\eta}\phi\exists\mu\in CF(\phi)\cdot(\mu\tran{\eta}\mathsf{True}\vee \mu\tran{\eta}\mu)$.
\end{lemma}
\begin{proof}
  According to Lemma \ref{lemma:formulacycle:exist} we know $\exists\mu\in CF(\phi)\cdot cl(\mu)\cap CF(\phi) =\{\mu\}$. Then we know for such $\mu$ it will meet and only meet $\mu\tran{\eta}\mathsf{True}\vee\mu\tran{\eta}\mu$ when each $\phi\tran{\eta}\phi$ holds.
\end{proof}

\begin{lemma}\label{lemma:cycle:partialorder}
  If $\phi\hookrightarrow\phi$, then there exists $S_0\subset S_1\subset\ldots\subset S_n= CF(\phi) (n\geq 0)$ such that $\forall \mu\in S_0\forall\phi\tran{\eta}\phi\cdot\mu\tran{\eta}\mathsf{True}\vee\mu\tran{\eta}\mu$, and for $i\geq 1$ we have $\forall\mu\in S_i\forall\phi\tran{\eta}\phi\cdot\mu\tran{\eta}\mu'$ and $CF(\mu')\subseteq S_{i-1}\cup\{\mu\}$.
\end{lemma}
\begin{proof}
  From Lemma~\ref{lemma:formulacycle:forall} we know $S_0\neq\emptyset$. Then let $S_1=S_0\cup \{\mu\mid\mu\in CF(\phi)\wedge\forall\phi\tran{\eta}\phi\cdot\mu\tran{\eta}\mu'\wedge CF(\mu')\subseteq S_0\cup\{\mu\}\}$. $S_1\supset S_0$ holds for the same reason with $S_0$ that $\exists\mu\in CF(\phi)\cdot cl(\mu)\cap CF(\phi) = S_0\cup\{\mu\}$, and such $\mu$s can be added into $S_1$. Inductively we can find the set $S_n=S_{n-1}\cup \{\mu\mid\mu\in CF(\phi)\wedge\forall\phi\tran{\eta}\phi\cdot\mu\tran{\eta_i}\mu'\wedge CF(\mu')\subseteq S_{n-1}\cup\{\mu\}\}$ ($n\geq 1$). Since $S_n\supset S_{n-1}$ and $|CF(\phi)|$ is limited and $\forall j\geq 0\cdot S_{j}\subseteq CF(\phi)$, so we can finally find $S_n = CF(\phi)$.
\end{proof}

\begin{figure}
\centering
  \includegraphics[scale = 0.6]{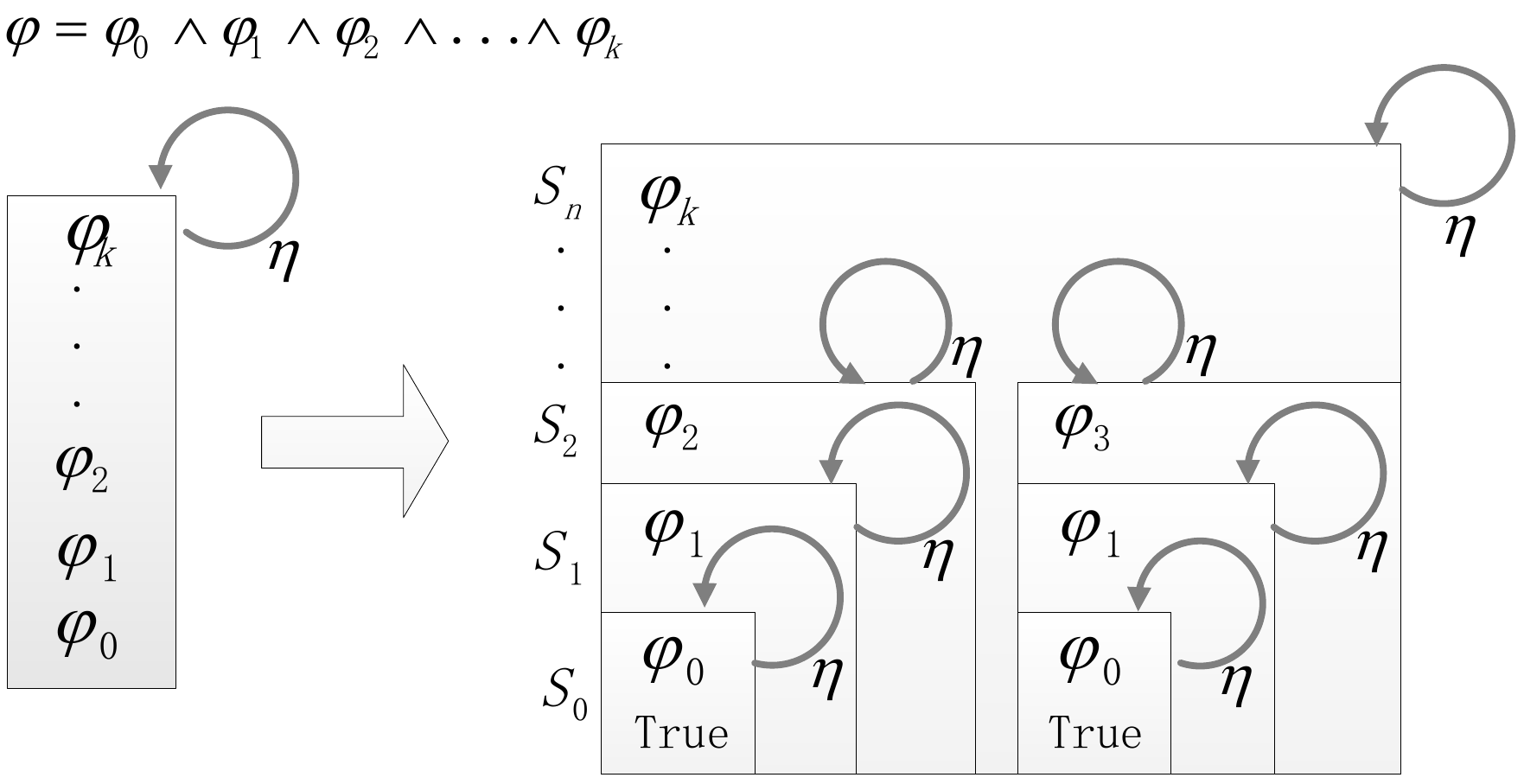}\\
  \caption{A demonstration of Lemma~\ref{lemma:cycle:partialorder} }\label{fig:loopingformula}
\end{figure}

A demonstration of this lemma is shown in Figure~\ref{fig:loopingformula}. In this case, $CF(\phi)=\{\phi_0,\phi_1,...,\phi_k\}$ and $\phi\tran{\eta}\phi$ holds. Then according to Lemma~\ref{lemma:cycle:partialorder} there exists $\phi_0$ so that $\phi_0\tran{\eta}\true\vee\phi_0\tran{\eta}\phi_0$ holds. Moreover, for $S_1=S_0\cup\{\phi_1\}$ we have $\phi\tran{\eta}\phi'$ and $CF(\phi')\subseteq S_0\cup\{\phi_1\}$. Note that including $S_{i-1}$ there can be more than one formulas added into $S_i$ at the same time: see $\phi_1$ and $\phi_3$ in $S_2$. This property for the looping formula plays a key role in the proofs in the following.

\begin{lemma}\label{lemma:finitesat:expandsat}
  $\phi\tran{\alpha}\psi\wedge S\models_f\psi\Rightarrow S\cup CF(\alpha)\models_f \phi$.
\end{lemma}
\begin{proof}
  We prove it by induction over the formula $\phi$.
  \begin{itemize}
    \item Basic step: If $\phi = p$, then we know $DNF(\phi)=\{p\wedge X\mathsf{True}\}$. So $\phi\tran{\alpha}\psi\Rightarrow p\in CF(\alpha)\wedge CF(\alpha)\models_f\mathsf{True}$. Thus $S\cup CF(\alpha)\models_f\phi$ is true.
    \item Inductive step: Assume $\phi_i (i=1,2)$ meet $\phi_i\tran{\alpha_i}\psi_i\wedge S_i\models_f\psi_i\Rightarrow S_i\cup CF(\alpha_i)\models_f \phi_i$, then we have:
        \begin{enumerate}
          \item If $\phi=X\phi_1$, then we know $DNF(\phi)=\{\mathsf{True}\wedge X(\phi_1)\}$. If $S\models_f\phi_1$ holds, then since $S\cup CF(\mathsf{True})\models_f\phi\equiv S\models_f\phi_1$, so $S\models_f \phi$ holds.
          \item If $\phi =\phi_1\vee\phi_2$, then we know $DNF(\phi)=DNF(\phi_1)\cup DNF(\phi_2)$, that is, $\forall \alpha\wedge X\psi\in DNF(\phi)\cdot \alpha\wedge X\psi\in DNF(\phi_1)\cup DNF(\alpha_2)$. If $S\models_f\psi$ holds then by induction hypothesis we have $\phi_1(\phi_2)\tran{\alpha}\psi\wedge S\models_f \psi\Rightarrow S\cup CF(\alpha)\models_f \phi_1(\phi_2)$, which indeed implies $S\cup CF(\alpha)\models_f \phi$ according to the definition of $\models_f$ (Definition~\ref{def:finitestepsat}). So $\phi\tran{\alpha}\psi\wedge S\models_f \psi\Rightarrow S\cup CF(\alpha)\models_f \phi$.
          \item If $\phi =\phi_1\wedge\phi_2$, then we know $\forall\alpha\wedge X\psi\in DNF(\phi)$ there exists $\alpha_i$ and $\psi_i(i=1,2)$ so that $\alpha=\alpha_1\wedge\alpha_2$ and $\psi=\psi_1\wedge\psi_2$ as well as $\alpha_1\wedge X\psi_1\in DNF(\phi_1)$ and $\alpha_2\wedge X\psi_2\in DNF(\phi_2)$. If $S\models_f\psi$ holds, then $S\models_f\psi_i(i=1,2)$ hold. By induction hypothesis we have $\phi_i\tran{\alpha_i}\psi_i\wedge S\models_f \psi_i\Rightarrow S\cup CF(\alpha_i)\models_f \phi_i (i=1,2)$, so $S\cup CF(\alpha_1)\cup CF(\alpha_2)\models_f \phi_1\wedge\phi_2$ holds. Thus $S\cup CF(\alpha)\models_f \phi$ holds.
          \item If $\phi =\phi_1 U\phi_2$, then we know for each $\alpha\wedge X\psi\in DNF(\phi)$, it is either in $DNF(\phi_2)$ or $\exists\alpha\wedge X\psi_1\in DNF(\phi_1)$ and $\psi = \psi_1\wedge \phi$. If $S\models_f \psi$ holds then $S\models_f \phi$ obviously holds when $\psi = \psi_1\wedge\phi$. Thus $S\cup CF(\alpha)\models_f \phi$ holds. And if $\alpha\wedge X\psi\in DNF(\phi_2)$ by induction hypothesis we have $S\cup CF(\alpha)\models_f\phi_2 \equiv S\cup CF(\alpha)\models_f \phi$ directly.
          \item If $\phi =\phi_1 R\phi_2$, then we know for each $\alpha\wedge X\psi\in DNF(\phi)$, it is either in $DNF(\phi_1\wedge\phi_2)$ or $\exists\alpha\wedge X\psi_2\in DNF(\phi_2)$ and $\psi = \psi_2\wedge \phi$. If $S\models_f \psi$ holds then we have proven $S\models_f \phi$ holds when $\alpha\wedge X\psi\in DNF(\phi_1\wedge\phi_2)$. And if $\psi = \psi_2\wedge\phi$ then $S\models_f \psi$ obviously makes $S\cup CF(\alpha)\models_f \phi$ hold.
        \end{enumerate}
  \end{itemize}
\end{proof}

\begin{lemma}\label{lemma:finitesat:expandsat2}
  Let $\phi_0=\phi\tran{\alpha_0}\phi_1\tran{\alpha_1}\phi_2\tran{\alpha_2}\ldots\tran{\alpha_n}\phi_{n+1}=\psi$ and $T=\bigcup_{0\leq j\leq n}\alpha_j$. If $S\models_f \psi$ then $S\cup T\models_f\phi$ holds.
\end{lemma}
\begin{proof}
  According to Lemma~\ref{lemma:finitesat:expandsat} we know $\phi_n\tran{\alpha_n}\phi_{n+1}=\psi\wedge S\models_f\psi\Rightarrow S\cup CF(\alpha_n)\models_f \phi_n$ holds. Inductively using Lemma~\ref{lemma:finitesat:expandsat} we can finally prove this lemma is true.
\end{proof}

\begin{lemma}\label{lemma:until:finitesat}
  If $\phi\tran{\eta}\phi$, then $\forall\mu\in UCF(\phi)\cdot\mu\tran{\eta}\mu'\wedge\mu\not\in CF(\mu')\Leftrightarrow\eta\models_f\phi$: here $UCF(\phi)\subseteq CF(\phi)$ and each $\mu$ in $UCF(\phi)$ is the Until formula.
\end{lemma}
\begin{proof}
  Let $\phi\tran{\eta}\phi= (\phi_0=\phi\tran{\omega_0}\phi_1\tran{\omega_1}\ldots\tran{\omega_{k}}\phi_{k+1}=\phi (k\geq 0))$ and the set $T=\bigcup_{0\leq j\leq k}\alpha_j$, where $\alpha_j\wedge X\phi_{j+1}\in DNF(\phi_j)\wedge\omega_j\vDash\alpha_j$ holds.

  ($\Rightarrow$). From Lemma~\ref{lemma:cycle:partialorder} we know there exists $S_0\subset S_1\subset\ldots\subset S_n= CF(\phi) (n\geq 0)$such that $\forall \mu\in S_0\forall\phi\tran{\eta}\phi\cdot\mu\tran{\eta}\mathsf{True}\vee\mu\tran{\eta}\mu$, and for $i\geq 1$ we have $\forall\mu\in S_i\forall\phi\tran{\eta}\phi\cdot\mu\tran{\eta}\mu'$ and $CF(\mu')\subseteq S_{i-1}\cup\{\mu\}$. For each $\mu$ in $S_0$, if $\mu\tran{\eta}\mathsf{True}$ then according to Lemma~\ref{lemma:finitesat:expandsat2} we have $T\models_f \mu$ holds; And if $\mu\tran{\eta}\mu$ since $\mu$ is not an Until formula, so $\mu$ is a Release formula. For the Release formula $\mu=\nu_1 R\nu_2$ we know every time $\mu\tran{\eta}\mu$ implies $\nu_2\tran{\eta}\mathsf{True}$. Thus according to Lemma~\ref{lemma:finitesat:expandsat2} we have $T\models_f\nu_2$ holds and then $T\models_f\mu$ holds according to its definition. So we prove now $\forall\mu\in S_0\cdot T\models_f\mu$. Inductively, for $i> 0$, if $\mu\in S_i$ and $\mu\tran{\eta}\mu'$ where $CF(\mu')\subseteq S_{i-1}$, and since we have proven $T\models_f\mu'$ then according to Lemma~\ref{lemma:finitesat:expandsat2} we know $T\models_f \mu$ holds. Else if $\mu\tran{\eta}\mu'\wedge\mu\in CF(\mu')$, then according to the assumption we know $\mu$ must be the Release formula, so for $\mu=\nu_1 R\nu_2$ we have $\nu_2\tran{\eta}\nu'$ where $CF(\nu')\subseteq S_{i-1}$. Since we have proven $T\models_f \nu'$ then according to Lemma~\ref{lemma:finitesat:expandsat2} we have $T\models_f \nu$ hold also. Then according to the definition of $\models_f$ we know $T\models_f\mu$ holds. Thus we can prove $\forall\mu\in S_n=CF(\phi)\cdot T\models_f\mu$, that is, $T\models_f\phi$ holds. Moreover since $\phi\tran{\eta}\phi$ is true thus according to the definition of $\models_f$ (Definition~\ref{def:finitestepsat}) we know $\eta\models_f \phi$ holds.

  ($\Leftarrow$) If $\exists\mu\in UCF(\phi)\cdot\mu\tran{\eta}\mu'\wedge \mu\in CF(\mu')$, then according to the expansion rule $\mu=\nu_1 U\nu_2=\nu_2\vee\nu_1\wedge X(\nu_1 U\nu_2)$ we can conclude $\eta\models_f\nu_2$ never holds, which makes $\eta\models_f\phi$ not hold. So the lemma is true.
\end{proof}

\begin{lemma}\label{lemma:release:sat}
  If $\phi$ is a $Realse$ formula, then $\phi\tran{\xi}\phi\Rightarrow\xi\vDash\phi$.
\end{lemma}
\begin{proof}
  Let $\phi = \mu R\nu$. Since $\phi\tran{\xi}\phi$, so we have $\exists n\cdot\phi\tran{\xi^n}\phi\wedge\phi\tran{\xi_n}\phi$. Let $\xi^n=\omega_0\omega_1\ldots\omega_n$ and $\eta_i = \omega_i\omega_{i+1}\ldots\omega_n (0\leq i\leq n)$. Thus we can easily know $\forall 0\leq i\leq n\cdot\nu\tran{\eta_i}\mathsf{True}$, which makes $\forall 0\leq j\leq n\cdot \xi_j\vDash\nu$. Inductively for $\phi\tran{\xi_n}\phi$ we can get the same conclusion. So $\forall j\geq 0$ we have $\xi\vDash\nu$, which makes $\xi\vDash\phi$ according to the LTL semantics.
\end{proof}

Now we begin to prove Lemma~\ref{lemma:finitesat:infsat}.
\begin{proof}
  From Lemma~\ref{lemma:cycle:partialorder} we know there exists $S_0\subset S_1\subset\ldots\subset S_n= CF(\phi) (n\geq 0)$such that $\forall \mu\in S_0\forall\phi\tran{\eta}\phi\cdot\mu\tran{\eta}\mathsf{True}\vee\mu\tran{\eta}\mu$, and for $i\geq 1$ we have $\forall\mu\in S_i\forall\phi\tran{\eta}\phi\cdot\mu\tran{\eta}\mu'$ and $CF(\mu')\subseteq S_{i-1}\cup\{\mu\}$.

   Basically for each $\mu$ in $S_0$, if $\exists\mu\tran{\eta_i}\mathsf{True}$ holds, then since $\forall 0\leq j\leq i\cdot\mu\tran{\eta_j}\mu$ so we have $\forall 0\leq j\leq i\cdot\xi'=\eta_j\eta_{j+1}\ldots\vDash\mu$; And if $\forall i\geq 0\cdot\mu\tran{\eta_i}\mu$, since $\eta_i\models_f\phi\Rightarrow \eta_i\models_f\mu$, and according to Lemma~\ref{lemma:until:finitesat} we know $\mu$ cannot be an Until formula. Then according to Corollary~\ref{coro:expand:cycle} we can know $\mu$ is a Release formula. Also we have $\mu\tran{\eta_i}\mu$, and according to Lemma~\ref{lemma:release:sat} we know $\forall i\geq 0\cdot\mu\tran{\eta_i}\mu$ plus $\mu$ is a release formula implies $\forall i\geq 0\cdot \xi'=\eta_i\eta_{i+1}\ldots\vDash\mu$. So first we can prove $\forall \mu\in S_0\forall i\geq 0\cdot\eta_i\eta_{i+1}\ldots\vDash\mu$.

  Inductively for the set $S_{n+1} (n \geq 0)$, if $\exists\mu\in S_{n+1}\setminus S_{n}\forall i\geq 0\cdot \mu\tran{\eta_i}\mu'\wedge CF(\mu')\subseteq S_{n}$, then from the basic step we know $\eta_{i+1}\eta_{i+2}\ldots\vDash\mu'$ so $\eta_i\eta_{i+1}\ldots\vDash\mu$. Moreover, we also have $\forall 0\leq j\leq i\cdot \eta_j\eta_{j+1}\ldots\vDash\mu$. If $\forall i\geq 0\cdot\mu\tran{\eta_i}\mu'\wedge\mu\in CF(\mu')$, similarly according to Lemma~\ref{lemma:until:finitesat} and Corollary~\ref{coro:expand:cycle} we know $\mu$ must be a Release formula. Let $\mu=\nu_1 R\nu_2$ and we know $\forall i\geq 0\cdot\nu_2\tran{\eta_i}\nu'\wedge CF(\nu')\subseteq S_n$. We have proven $\eta_{i+1}\eta_{i+2}\ldots\vDash\nu'$, so we have $\forall i\geq 0\cdot \eta_i\eta_{i+1}\ldots\vDash\nu_2$. Then according to the LTL semantics we have $\forall i\geq 0\cdot\eta_i\eta_{i+1}\ldots\vDash\mu$. So we can prove now $\forall\mu\in S_{n+1}\forall i\geq 0\cdot\eta_i\eta_{i+1}\ldots\vDash\mu$.

  So finally we can prove the set $S_n= CF(\phi)$, $\forall\mu\in S_n\forall i\geq 0\cdot\eta_i\eta_{i+1}\ldots\vDash\mu$. Thus $\forall \mu\in S_n\forall i\geq 0\cdot\eta_i\eta_{i+1}\ldots\vDash\mu$ implies $\xi\vDash \phi$.
\end{proof}

\subsection{Proof of Lemma \ref{lemma:completeness}}
\begin{lemma}\label{lemma:complete:finiteexist}
  $\xi\vDash\phi\Rightarrow\exists n\cdot \xi^n\models_f\phi$.
\end{lemma}
\begin{proof}
  We prove it by induction over the size of formula $\phi$.
  \begin{itemize}
    \item Basic step: If $\phi=p$, then $\xi\vDash\phi\Rightarrow p\in \xi^1$. So according to Definition \ref{def:finitestepsat} we know $\xi^1\models_f \phi$ is true.
    \item Inductive step: Assume for the formulas $\phi_i$ ($i=1,2$) we have $\xi\vDash\phi_i\Rightarrow\exists n\cdot \xi^n\models_f\phi_i$ hold. Then
        \begin{enumerate}
          \item If $\phi = X\phi_1$, then $\xi\vDash\phi\Rightarrow\xi_1\vDash\phi_1$. By induction hypothesis we know $\exists n\cdot {\xi_1}^n\models_f\phi_1$ holds, so ${\xi}^{n+1}\models_f\phi$ also holds.
          \item If $\phi = \phi_1\wedge\phi_2$, then $\xi\vDash\phi\Rightarrow\xi\vDash\phi_1\wedge\xi\vDash\phi_2$. By induction hypothesis we know $\exists n_1\cdot {\xi}^{n_1}\models_f\phi_1$ and $\exists n_2\cdot {\xi}^{n_2}\models_f\phi_2$ hold, so we can conclude $\exists n\geq max(n_1,n_2)\cdot {\xi}^{n}\models_f\phi$ holds.
          \item If $\phi = \phi_1\vee\phi_2$, then $\xi\vDash\phi\Rightarrow\xi\vDash\phi_1\vee\xi\vDash\phi_2$. By induction hypothesis we know $\exists n_1\cdot {\xi}^{n_1}\models_f\phi_1$ or $\exists n_2\cdot {\xi}^{n_2}\models_f\phi_2$ hold, so we can conclude $\exists n=n_1\vee n=n_2\cdot {\xi}^{n}\models_f\phi$ holds.
          \item If $\phi = \phi_1 U \phi_2$, then $\xi\vDash\phi_1 U\phi_2\Rightarrow\exists i\geq 0\cdot\xi_i\vDash\phi_2$. By induction hypothesis we have $\exists n\cdot {\xi_i}^n\models_f\phi_2$ hold, so from Definition \ref{def:finitestepsat} we have $\xi^{i+n}\models_f \phi$ hold.
          \item If $\phi = \phi_1 R \phi_2$, then $\xi\vDash\phi_1 R\phi_2\Rightarrow\forall i\geq 0\cdot\xi_i\vDash\phi_2$. So $\xi\vDash\phi_2$ holds. Then By induction hypothesis we know $\exists n\cdot \xi^n\models_f\phi_2$ and according to Definition \ref{def:finitestepsat} we know $\xi^n\models_f\phi$ also holds.
        \end{enumerate}
  \end{itemize}
\end{proof}

\begin{lemma}\label{lemma:complete:recursive}
  $\phi\tran{\xi}\phi\wedge\xi\vDash\phi\Rightarrow\exists n\cdot\phi\tran{\xi^n}\phi\wedge \xi^n\models_f\phi\wedge (\phi\tran{\xi_n}\phi\wedge\xi_n\vDash\phi)$.
\end{lemma}
\begin{proof}
  We first prove $\phi\tran{\xi}\phi\wedge\xi\vDash\phi\Rightarrow\exists n\cdot\phi\tran{\xi^n}\phi\wedge \xi^n\models_f\phi$. If $\forall n\cdot\phi\tran{\xi^n}\phi\wedge \neg (\xi^n\models_f\phi)$, we can conclude $\forall i\leq n\cdot \neg (\xi^i\models_f\phi)$, thus causing the contradiction with Lemma~\ref{lemma:complete:finiteexist}. Moreover, since $\phi\tran{\xi}\phi\wedge\phi\tran{\xi^n}\phi\wedge\xi\vDash\phi$, so $\phi\tran{\xi_n}\phi\wedge\xi_n\vDash\phi$ is also true. So this lemma is true.
\end{proof}

  To prove Lemma~\ref{lemma:completeness} we can use Lemma~\ref{lemma:complete:recursive} inductively, and obviously it is true.

\subsection{Proof of Theorem \ref{thm:central}}
\begin{proof}
  ($\Rightarrow$). According to Corollary~\ref{coro:expand:existcycle} and Lemma~\ref{lemma:completeness} we know it is true.

  \hspace*{5mm}($\Leftarrow$). According to Corollary~\ref{coro:expand:existcycle} and Lemma~\ref{lemma:finitesat:infsat} it is true.
\end{proof}

\subsection{Proof of Theorem \ref{thm:correct}}
\begin{lemma}\label{lemma:automata:runable}
  Let $\xi=\omega_0\omega_1\ldots$ and $\mathcal{A}_{\lambda}$ the B\"uchi automaton for $\lambda$ generated by \textit{DNF-based} construction. Then $\psi_0=\lambda\tran{\omega_0}\psi_1\tran{\omega_1}\ldots\tran{\omega_{n-1}}\psi_n$ holds, where $\psi_i\in EF(\lambda)$, if and only if there is a corresponding path $s_0\tran{\omega_0}s_1\tran{\omega_1}\ldots\tran{\omega_{n-1}}s_n$ in $\mathcal{A}_{\lambda}$ where each $s_i$ is the $\psi_i$-state.
\end{lemma}
\begin{proof}
We prove it by induction over $n$.

1). When $n = 1$, if $\psi_1\in DNF(\lambda)$, then according to our construction directly we know for $\psi_0=\lambda\tran{\omega_0}\psi_1$, if and only if there is a $s_0\tran{\omega_0}s_1$ where $s_i$ is the $\psi_i$-state and $\lambda\tran{\omega_0}\psi_1$.

2). When $n = k, k\geq 1$ we assume $\psi_0=\lambda\tran{\omega_0}\psi_1\tran{\omega_1}\ldots\tran{\omega_{k-1}}\psi_k$ if and only if there is a corresponding path $s_0\tran{\omega_0}s_1\tran{\omega_1}\ldots\tran{\omega_{k-1}}s_k$ where for $k \geq i\geq 0$ each $s_i$ is the $\psi_i$-state in $\mathcal{A}_{\lambda}$. Then for $\psi_0=\lambda\tran{\omega_0}\psi_1\tran{\omega_1}\ldots\tran{\omega_{k-1}}\psi_k\tran{\omega_{k}}\psi_{k+1}$ holds, we know if and only if $\exists\alpha_k\wedge X\psi_{k+1}\in DNF(\psi_k)\wedge\omega_k\models\alpha_k$ holds from Definition~\ref{def:expand}. According to the construction we know $\exists\alpha_k\wedge \psi_{k+1}\in DNF(\psi_k)\wedge\omega_k\models\alpha_k$ if and only if there is a $s_k\tran{\omega_k}s_{k+1}$ where $s_{k+1}$ is $\psi_{k+1}$-state. So it is true that $\psi_0=\lambda\tran{\omega_0}\psi_1\tran{\omega_1}\ldots\tran{\omega_{k-1}}\psi_k\tran{\omega_{k}}\psi_{k+1}$ if and only if there is a $s_0\tran{\omega_0}s_1\tran{\omega_1}\ldots\tran{\omega_{k-1}}s_k\tran{\omega_k}s_{k+1}$ in $\mathcal{A}_{\lambda}$. The proof is done.
\end{proof}


Now we come to prove Theorem~\ref{thm:central}.
\begin{proof}
  ($\Leftarrow$) Let $\xi=\omega_0\omega_1\ldots$ be an accepting run of
  $\A_\lambda$, and we want to prove that $\xi\models\lambda$. Let
  $\sigma:=s_0\tran{\omega_0}s_1\tran{\omega_1}\ldots$ be the
  corresponding path accepting $\xi$. Thus, $inf(\sigma)$ contains at
  least one accepting state $s\in F$. Assume $s=\langle\phi, \emptyset\rangle$. Since there exists a finite path $s_0\tran{\omega_0}s_1\tran{\omega_1}s_2\ldots\tran{\omega_n}s_{n+1}=s$, where each $s_i$ is the $\phi_i$-state. According to Lemma~\ref{lemma:automata:runable} we know $\lambda=\phi_0\tran{\omega_0}\phi_1\tran{\omega_1}\phi_2\ldots\tran{\omega_n}\phi_{n+1}=\phi$ holds. Then we know $\exists \xi_n=\eta_1\eta_2\ldots$ so that for each $\eta_i=\omega_{i_0}\omega_{i_1}\ldots\omega_{i_n} (i,n\geq 1)$ we have $s_{i_0}=s\tran{\omega_{i_0}}s_{i_1}\tran{\omega_{i_1}}\ldots\tran{\omega_{i_n}}s_{i_{n+1}}=s$, of which for simplicity we denote as $s\tran{\eta_i}s$. According to Lemma~\ref{lemma:automata:runable} we know each time $s\tran{\eta_i}s$ holds $\phi\tran{\eta_i}\phi$ also holds ($s$ is the $\phi$-state). Moreover, according to our construction and Lemma~\ref{lemma:obligaionandsatonce} we know $\eta_i\models_f\phi$ holds. Finally according to Theorem~\ref{thm:central} we can conclude $\xi\vDash\lambda$.

($\Rightarrow$) Let $\xi=\omega_0\omega_1\ldots$ and $\xi\vDash\lambda$, we now prove there is an accepting run $\sigma=s_0\tran{\omega_0}s_1\tran{\omega_1}\ldots$ in $\mathcal{A}_{\lambda}$. From Theorem~\ref{thm:central} we know $\xi\vDash\lambda\Rightarrow\exists\phi\exists n\cdot\lambda\tran{\xi^n}\phi\wedge(\exists\xi_n=\eta_1\eta_2\ldots\cdot\forall i\geq 1\cdot\phi\tran{\eta_i}\phi\wedge \eta_i\models_f\phi)$. According to Lemma~\ref{lemma:automata:runable} we can find an infinite path $\sigma = s_0\tran{\omega_0}s_1\tran{\omega_1}\ldots\tran{\omega_{n-1}}s\tran{\omega_n}\ldots$ in $\mathcal{A}_{\lambda}$ on which $\xi$ can run. Here $s_0$ is the $\lambda$-state and $s$ is the $\phi$-state, and for each $\eta_i=\omega_{i_0}\omega_{i_1}\ldots\omega_{i_n} (i,n\geq 1)$ we have $s_{i_0}=s\tran{\omega_{i_0}}s_{i_1}\tran{\omega_{i_1}}s_{i_2}\ldots\tran{\omega_{i_n}}s_{i_{n+1}}=s$, of which for simplicity we denote as $s\tran{\eta_i}s$. Let $s_{i_j} (n+1 \geq j\geq 0)$ be the $\phi_{i_j}$-state, and the set $T=\bigcup_{0\leq k\leq n}\alpha_{i_k}$ where each $\alpha_{i_k}$ satisfies $\exists\alpha_{i_k}\wedge X(\phi_{i_{k+1}})\in DNF(\phi_{i_k})\wedge \omega_{i_k}\models\alpha_{i_k}$. Since $T\models_f \phi$ holds so according to Lemma~\ref{lemma:obligaionandsatonce} we know $\exists O\in OS_{\phi}\cdot O\subseteq T$. Moreover, our construction guarantees for each $s\tran{\eta_i}s$ there is $s_{i_j}=\langle\phi_{i_j}, P\rangle (0\leq j\leq n)$ so that $P=\emptyset$. Since such states with the format of $\langle-, \emptyset\rangle$ is finite, so there must be such a state in $inf(\sigma)$. Finally we prove the theorem is true.
\end{proof}

\end{document}